\documentclass[12pt,journal,onecolumn,draftclsnofoot]{IEEEtran}

\ifCLASSINFOpdf
\else
\fi
\usepackage{amsthm}
\usepackage{graphicx}
\usepackage{amsmath}
\usepackage{cite}
\usepackage{amssymb}
\usepackage{multirow}
\usepackage{color}
\usepackage{epstopdf}
\usepackage{float}

\usepackage{flushend}
\usepackage{algorithm}
\usepackage{algorithmic}
\usepackage{tabularx}

\usepackage{multirow} 
\usepackage{hhline}

\newtheorem{theorem}{Theorem}
\newtheorem{lemma}[theorem]{Lemma}

\begin{document}

%
\title{Antenna Elements' Trajectory Optimization for Throughput Maximization in Continuous-Trajectory Fluid Antenna-Aided Wireless Communications
\thanks{
	
	S. Yang, Y. Li and Y. Xiao are with the National Key Laboratory of Wireless Communications, University of Electronic Science and Technology of China, Chengdu 611731, China (e-mail: shuaixin.yang@foxmail.com, nikoeric@foxmail.com, xiaoyue@uestc.edu.cn).
	
	Y. Guan and C. Yuen is with the School of Electrical and Electronic Engineering, Nanyang Technological University, Singapore (e-mail: EYLGuan@ntu.edu.sg, chau.yuen@ntu.edu.sg.).
	
	K.-K. Wong is with the Department of Electronic and Electrical Engineering, University College London, WC1E 7JE London, U.K., and also with the Department of Electronic Engineering, Kyung Hee University, Yongin-si, Gyeonggi-do 17104, Republic of Korea (e-mail: kai-kit.wong@ucl.ac.uk).
	
	H. Shin is with the Department of Electronics and Information Convergence
	Engineering, Kyung Hee University, 1732 Deogyeong-daero, Giheung-gu,
	Yongin-si, Gyeonggi-do 17104, Republic of Korea (e-mail: hshin@khu.ac.kr).
}
}

\author{Shuaixin~Yang,~\IEEEmembership{}
		Yijia~Li,~\IEEEmembership{}
		Yue~Xiao,~\IEEEmembership{}
        Yong Liang Guan,~\IEEEmembership{}
        Kai-Kit Wong,~\IEEEmembership{}
        Hyundong Shin,~\IEEEmembership{}
        and Chau Yuen~\IEEEmembership{}
}

\maketitle

\begin{abstract}
Fluid antenna (FA) systems offer novel spatial degrees of freedom (DoFs) with the potential for significant performance gains. Compared to existing works focusing solely on optimizing FA positions at discrete time instants, we introduce the concept of continuous-trajectory fluid antenna (CTFA), which explicitly considers the antenna element's movement trajectory across continuous time intervals and incorporates the inherent kinematic constraints present in practical FA implementations. Accordingly, we formulate the total throughput maximization problem in CTFA-aided wireless communication systems, addressing the joint optimization of continuous antenna trajectories in conjunction with the transmit covariance matrices under kinematic constraints. To effectively solve this non-convex problem with highly coupled optimization variables, we develop an iterative algorithm based on block coordinate descent (BCD) and majorization-minimization (MM) principles with the aid of the weighted minimum mean square error (WMMSE) method. Finally, numerical results are presented to validate the efficacy of the proposed algorithms and to quantify the substantial total throughput advantages afforded by the conceived CTFA-aided system compared to conventional fixed-position antenna (FPA) benchmarks and alternative approaches employing simplified trajectories.
\end{abstract}


\IEEEpeerreviewmaketitle


\section{Introduction}

Sixth-generation (6G) wireless networks represent an envisioned paradigm shift, transcending conventional communication frameworks to deliver ultra-reliable, low-latency, and intelligent connectivity across heterogeneous and dynamic environments, including terrestrial, aerial, and spaceborne domains \cite{6G1}. Achieving this ambitious vision critically hinges upon fundamental physical layer innovations, necessitating unprecedented levels of reconfigurability, adaptability, and deeply integrated artificial intelligence (AI) capabilities \cite{6G2,6G5}. Among the key enabling technologies is the fluid antenna (FA), a novel architecture offering the capability to dynamically reconfigure the antennas' spatial position physically within a defined region. Such unique characteristic unlocks new degrees of freedom (DoFs), critically enhancing the potential for spatial diversity exploitation and communication performance optimization \cite{6G3, 6G4, FAS1,analytical,CE}.

This fundamental capability of FA has spurred significant research interest. Building upon the pioneering insights presented in \cite{FAS1}, recent research has further unveiled the versatility and significant advantages of FA systems over conventional fixed-position antenna (FPA) architectures, primarily by exploiting dynamic antenna position adjustment. 
In particular, within the realm of physical layer security (PLS), the inherent spatial DoFs of FA have facilitated advanced security strategies.
For instance, joint optimization of FA positions and transmit beamforming effectively maximizes secrecy rates, even under strict covertness constraints \cite{PLS1}. 
Furthermore, novel coding-enhanced cooperative jamming techniques ensure interference is cancelable solely by the legitimate receiver, enhancing secrecy via optimized port selection and power control \cite{PLS2}. 
Synergistic combinations of FA with reconfigurable intelligent surfaces (RIS) have also demonstrated significant security augmentation, often evaluated via secrecy outage probability (SOP) analysis \cite{PLS3,RIS1,RIS2}. 
Complementing these designs, rigorous analyses under correlated fading derive compact expressions for average secrecy capacity (ASC), SOP, and secrecy energy efficiency (SEE), revealing FA's advantages over traditional antenna diversity schemes \cite{PLS4}. 
These benefits are further evidenced in complex multi-user non-orthogonal multiple access (NOMA)-based wireless powered communication networks (WPCNs), quantifying ASC and SOP improvements despite external and internal eavesdroppers \cite{PLS5}.

Beyond security-oriented applications, FA has demonstrated its ability to enhance performance across various communication paradigms. For example, within cognitive radio (CR) networks, FA-equipped secondary users can adaptively adjust their antenna positions to achieve superior spectrum sensing accuracy, effectively addressing stringent false-alarm constraints \cite{CR}. Similarly, in backscatter communication scenarios, FA deployed at the reader facilitates improved link reliability and reduces outage probability and delay outage rates by exploiting spatial DoFs \cite{BackScatter}. Furthermore, integrating FA into wideband orthogonal frequency-division multiplexing (OFDM) systems, exemplified by 5G New Radio (NR), has shown notable potential for throughput enhancement via adaptive modulation and strategic port selection tailored to frequency-selective channel characteristics \cite{5GNR}. Within the emerging domain of integrated sensing and communication (ISAC), FA technology provides novel avenues for managing the fundamental sensing-communication trade-off via dynamic antenna configuration. 
Specifically, jointly optimizing FA positions and transmit precoding has been shown to maximize radar sensing signal-clutter-noise ratio (SCNR) under communication signal-to-interference-plus-noise ratio (SINR) constraints \cite{ISAC1}. 
Furthermore, FA enables substantial transmit power minimization while satisfying both communication and sensing requirements, demonstrating flexibility in shifting the ISAC operational point \cite{ISAC2}. 
Extending this adaptability, jointly optimizing beamforming and FA positions at both the base station and user sides effectively maximizes downlink communication rates subject to stringent sensing gain requirements \cite{ISAC3}.

Despite the substantial body of research highlighting the theoretical advantages and diverse applications of FA systems, there remains a notable gap between the predominantly adopted theoretical models and the practical realities of certain crucial hardware implementations, particularly those based on mechanical actuation, such as motor-driven systems \cite{motor1,motor2}. The existing optimization frameworks and performance analyses typically assume an idealized scenario characterized by \emph{instantaneous and cost-free switching} among a set of predefined antenna port positions. Although this assumption provides significant analytical convenience, it nevertheless neglects critical physical dynamics intrinsic to mechanically actuated antenna systems.
A potential consequence is that performance predictions derived from such simplified models may deviate from those observed in practice, and optimization algorithms conceived under these assumptions might require adaptation for deployment on motor-driven FA platforms. 

More specifically, motor-based FA architectures, which involve the physical translation of antenna elements (AEs), introduce unique characteristics, notably the capability for \emph{continuous} spatial movement. This property enables finer spatial sampling granularity and higher spatial resolution, thereby facilitating sophisticated trajectory-based optimization strategies wherein antenna positions dynamically trace optimized trajectories corresponding to evolving channel conditions. Nevertheless, the inherent physical movements of such systems impose substantial dynamic factors fundamentally incompatible with the idealized instantaneous switching paradigm. These critical dynamics encompass significant \textbf{motion delay} (the finite duration required for antenna translation) and additional \textbf{mechanical constraints} (limitations on maximum velocity, acceleration, and available range of motion).

Against this backdrop, several insightful studies have begun to incorporate mobility aspects. For instance, Ref. \cite{trajectory1} primarily optimized trajectories to minimize physical reconfiguration time between positions subject to kinematic constraints, potentially neglecting trajectory optimization during active data transmission for communication enhancement. Additionally, Ref. \cite{trajectory2} investigated mobile antennas on unmanned aerial vehicles (UAVs), which generally pertains to macro-scale platform mobility rather than the micro-scale element repositioning germane to FAs. Furthermore, the work in \cite{6DMA} explored six-dimensional movable antennas (6DMAs), where the static three-dimensional (3D) positions and rotations of antenna surfaces are optimized to align with the long-term spatial user distribution, which is distinct from optimizing an antenna's continuous trajectory during active transmission to leverage instantaneous channel dynamics. Therefore, despite these initial efforts, a comprehensive framework fully capturing the physics of continuously moving antennas and facilitating joint communication-mobility optimization warrants further investigation.

In this contribution, we specifically address FA systems capable of continuous physical movement by introducing the concept of continuous-trajectory fluid antenna (CTFA), explicitly incorporating realistic kinematic constraints intrinsic to practical implementations. 
Within this CTFA framework, we formulate the total throughput maximization problem, aiming at jointly optimizing the AEs' continuous trajectory and the associated transmit covariance matrix (TCM) throughout the communication session. 
However, as the AEs move along their continuous trajectories, the system may inevitably encounter deleterious deep fading conditions, while conversely identifying advantageous positions characterized by the constructive superposition of multipath components. 
Therefore, the throughput-maximization trajectory design needs to effectively strike an optimal balance between circumventing deep channel nulls and exploiting favourable propagation positions.

In a nutshell, the primary contributions of this paper are summarized as follows:
\begin{itemize}
	\item We formulate a novel joint antenna trajectory and TCM optimization problem under the CTFA framework, explicitly incorporating practical channel characteristics and kinematic constraints during continuous motion.
	\item We develop an efficient alternating optimization algorithm, integrating block coordinate descent (BCD) and majorization-minimization (MM) principles, to find high-quality solutions for the formulated non-convex AE trajectory optimization problem.
	\item We demonstrate through numerical simulations that the proposed CTFA framework, accounting for realistic physical dynamics, yields significant performance improvements compared to conventional FPA systems and alternative approaches employing simplified trajectories.
\end{itemize}

The remainder of this paper is structured as follows. Section II introduces the system model for CTFA systems and formulates the optimization problem for throughput maximization. Section III details the proposed BCD-MM-based algorithm. Section IV provides the simulation results to validate the proposed framework. Finally, Section V concludes the paper.

\begin{figure*}[htbp]
	\centering
	\includegraphics[width=7in]{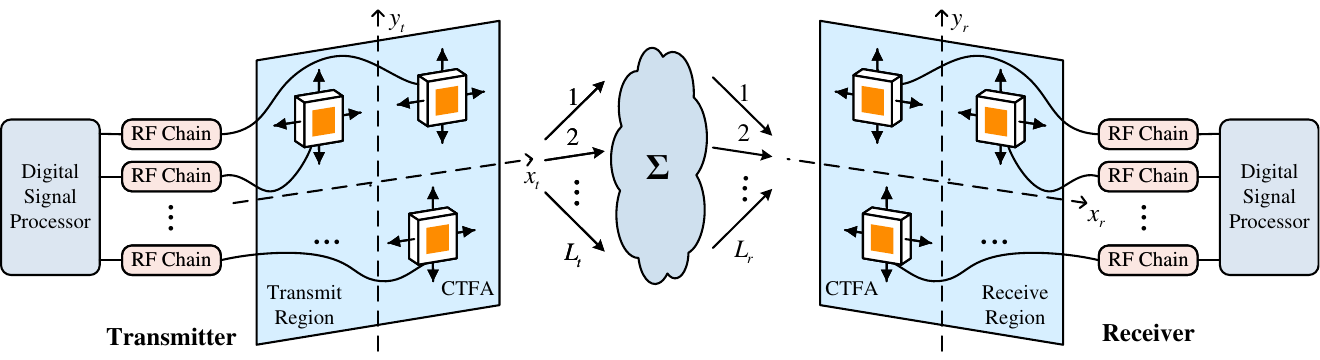}
	\caption{System model of the point-to-point CTFA-enabled MIMO system.}
	\label{SystemModel}
\end{figure*}

\emph{Notation}: Unless otherwise specified, bold lowercase and uppercase letters are used for column vectors and matrices, respectively. $\mathbb{C}$ denotes the set of complex numbers. $(\cdot)^T$ and $(\cdot)^H$ denote the transposition and Hermitian conjugate transposition, respectively. $\mathrm{tr}(\cdot)$, $\det(\cdot)$, and $\mathrm{rank}(\cdot)$ denote the matrix trace, determinant, and rank operations, respectively. $\|\cdot\|_2$ represents the Euclidean-2 norm for vectors. $\mathrm{Re}\{\cdot\}$ takes the real part of a complex number. $\nabla$ and $\nabla^2$ denote the gradient and Hessian operators, respectively. $\lambda_{\max}(\cdot)$ denotes the maximum eigenvalue of a matrix. $\mathrm{diag}(\mathbf{x})$ returns a diagonal matrix with the elements of vector $\mathbf{x}$ on its main diagonal. $\mathbf{I}$ represents the identity matrix and $X \setminus Y$ denotes the set of elements belonging to $X$ but not to $Y$. $\mathbf{A} \succeq \mathbf{0}\left(\mathbf{A}\succ\mathbf{0} \right) $ denotes that the matrix $\mathbf{A}$ is positive semidefinite (definite).

\section{System Model and Problem Formulation}\label{}

\subsection{System Model }

Fig.~\ref{SystemModel} portrays the system model conceived for a point-to-point multiple-input multiple-output (MIMO) system incorporating CTFA. 
Specifically, the transmitter and receiver are equipped with $N_t$ and $N_r$ AEs, respectively. 
These elements are linked to the associated radio-frequency (RF) chains via flexible connectors, which facilitate the dynamic adjustment of antenna positions upon dedicated panels actuated by micro motors. 

For CTFA systems, the continuous antenna movement inherently induces time-varying antenna positions, which consequently gives rise to time-varying channel characteristics. It is pertinent to emphasize that our analysis is predicated upon a quasi-static channel model, which is readily satisfied, for instance, if the operational duration is confined within the channel coherence interval or in scenarios where both the environmental scatterers and the transceiver platforms remain stationary relative to the propagation environment. 

Therefore, we formulate a throughput expression that explicitly incorporates the continuous temporal antenna displacement, thereby setting the stage for the optimization strategies detailed subsequently. Firstly, we respectively denote $\mathbf{Q}\left(\tau\right) \in {\mathbb C}^{N_t\times N_t}$ and ${\mathbf H}\left(\tau\right)\in {\mathbb C}^{N_r\times N_t}$ as the time-varying TCM and the channel matrix of CTFA systems to be elaborated in the next subsection. Then, the instantaneous achievable rate can be expressed as
\begin{equation}\label{1}
	R\left(\tau\right)={\log}\det \left( {{{\bf{I}}} + \frac{1}{{{\sigma^2}}}{{\bf H}\left(\tau\right)}{\bf Q}\left(\tau\right){{{\bf H}}\left(\tau\right)^H}} \right),
\end{equation}
where ${\mathbf Q}\left(\tau\right)$ satisfies the power constraint ${\rm tr}\left({\mathbf Q}\left(\tau\right)\right)\leq P$ and $\sigma^2$ denotes the variance of the additive white Gaussian noise (AWGN) at the receiver.

Thus, the total throughput $C$ that can be transmitted from the transmitter to the receiver over the duration $T$ is expressed as
\begin{equation}\label{2}
	C=\int_0^T{\log}\det \left( {{{\bf{I}}} + \frac{1}{{{\sigma^2}}}{{\bf H}\left(\tau\right)}{\bf Q}\left(\tau\right){{{\bf H}}\left(\tau\right)^H}} \right)d\tau.
\end{equation}

\subsection{Channel Model}
The channel characterizing the CTFA system is modeled based on the far-field multipath field-response framework detailed in \cite{Capacity,SM}. 
Explicitly, let $L_t$ and $L_r$ denote the number of dominant propagation paths associated with the transmitter and receiver, respectively. 
The propagation delay incurred along the $p$-th transmit path ($p=1,2,\dots,L_t$) is intrinsically dependent on the instantaneous positions of the transmit AEs, which is represented as ${\mathbf q}_t\left(\tau\right) = [x_t\left(\tau\right), y_t\left(\tau\right)]^T$. 
Moreover, the AEs are constrained to move within the square region $\mathcal{C}_t = [0,A] \times [0,A]$, with the origin specified as $\mathbf{o}_t = [0,0]^T$. 
Consequently, for the $p$-th transmit path, the position-dependent propagation offset $\rho_t^p({\mathbf q}_t\left(\tau\right))$ relative to the origin $\mathbf{o}_t$ is formulated as 
\begin{equation}\label{3}
	\rho_t^p\left(\mathbf{q}_t\left(\tau\right)\right) = x_t\left(\tau\right) \sin\theta_t^p \cos\phi_t^p + y_t\left(\tau\right) \cos\theta_t^p, 
\end{equation}
where $\theta_t^p$ and $\phi_t^p$ denote the corresponding elevation and azimuth angles of departure (AoDs), respectively.

Then, the phase shift incurred by the $p$-th transmit path is formulated as $2\pi \rho_t^p(\mathbf{q}_t\left(\tau\right)) / \lambda$, where $\lambda$ denotes the signal wavelength. 
Leveraging this phase shift formulation, the transmit field response vector $\mathbf{g}(\mathbf{q}_t\left(\tau\right))$ associated with each transmit AE is constructed as
\begin{equation}\label{4}
	\mathbf{g}\left(\mathbf{q}_t\left(\tau\right)\right) \triangleq \left[e^{j\frac{2\pi}{\lambda}\rho_t^1(\mathbf{q}_t\left(\tau\right))}, e^{j\frac{2\pi}{\lambda}\rho_t^2(\mathbf{q}_t\left(\tau\right))}, \dots, e^{j\frac{2\pi}{\lambda}\rho_t^{L_t}(\mathbf{q}_t\left(\tau\right))}\right]^T.
\end{equation}
Furthermore, letting $\tilde{\mathbf{q}}_t\left(\tau\right) \triangleq \left\{\mathbf{q}_t^k\left(\tau\right)\right\}_{k=1}^{N_t}$ denote the collection of positions for all AEs at the BS, the composite field response matrix $\mathbf{G}(\tilde{\mathbf{q}}_t\left(\tau\right))$ incorporating all $N_t$ transmit AEs is defined as
\begin{equation}\label{5}
	\mathbf{G}(\tilde{\mathbf{q}}_t\left(\tau\right))\triangleq\left[\mathbf{g}\left(\mathbf{q}_t^1\left(\tau\right)\right),\mathbf{g}\left(\mathbf{q}_t^2\left(\tau\right)\right),\ldots,\mathbf{g}\left(\mathbf{q}_t^{N_t}\left(\tau\right)\right)\right].
\end{equation}

Correspondingly, at the receiver side, the elevation and azimuth angles of arrival (AoAs) associated with the $q$-th ($q=1,2,\dots,L_r$) receive path are denoted by $\theta_r^q \in [0, \pi]$ and $\phi_r^q \in [0, \pi]$, respectively. 
The resultant field response vector $\mathbf{f}(\mathbf{q}_r\left(\tau\right))$ for each receive AE positioned at ${\mathbf q}_r\left(\tau\right) = [x_r\left(\tau\right), y_r\left(\tau\right)]^T \in \mathcal{C}_r = [0, A] \times [0, A]$, is given by
\begin{equation}\label{6}
	\mathbf{f}(\mathbf{q}_r\left(\tau\right))\triangleq\left[e^{j\frac{2\pi}{\lambda}\rho_r^1\left(\mathbf{q}_r(\tau)\right)},e^{j\frac{2\pi}{\lambda}\rho_r^2\left(\mathbf{q}_r(\tau)\right)},\ldots,e^{j\frac{2\pi}{\lambda}\rho_r^{L_r}(\mathbf{q}_r(\tau))}\right]^T,
\end{equation}
where $\rho_r^q(\mathbf{q}_r\left(\tau\right)) = x_r\left(\tau\right) \sin\theta_r^q \cos\phi_r^q + y_r\left(\tau\right) \cos\theta_r^q$.
Subsequently, the corresponding receive field response matrix is expressed as
\begin{equation}\label{7}
	\mathbf{F}(\tilde{\mathbf{q}}_r\left(\tau\right))\triangleq\left[\mathbf{f}\left(\mathbf{q}_r^1\left(\tau\right)\right),\mathbf{f}\left(\mathbf{q}_r^2\left(\tau\right)\right),\ldots,\mathbf{f}\left(\mathbf{q}_r^{N_r}\left(\tau\right)\right)\right],
\end{equation}
where $\tilde{\mathbf{q}}_r\left(\tau\right) \triangleq \left\{\mathbf{q}_r^l\left(\tau\right)\right\}_{l=1}^{N_r}$. Furthermore, the path response matrix linking the reference points of the transmit and receive regions, $\mathbf{o}_t$ and $\mathbf{o}_r$, is denoted by $\mathbf{\Sigma} \in \mathbb{C}^{L_r \times L_t}$. 
Its elements, which capture the field response between the pertinent transmit and receive paths, are modeled as time-invariant independent and identically distributed (i.i.d.) complex random variables due to the assumption of a quasi-static channel. 
Additionally, we postulate a rich scattering environment, implying that $\min\left(L_t, L_r\right) \geq \max\left(N_t, N_r\right)$. 
Based on the preceding definitions, the end-to-end MIMO channel matrix ${\mathbf H}\left(\tau\right)$ characterizing the link between the transmitter and receiver equipped with CTFAs can finally be formulated as
\begin{equation}\label{8}
	{\mathbf H}\left(\tau\right)=\mathbf{F}(\tilde{\mathbf{q}}_r\left(\tau\right))^H\mathbf{\Sigma} \mathbf{G}(\tilde{\mathbf{q}}_t\left(\tau\right)).
\end{equation}
Evidently, the channel's temporal variations principally stem from the dynamic movement of the AEs.

\subsection{Kinematic Constraints}
In this subsection, we explicitly consider the kinematic constraints imposed on the CTFAs. 
Specifically, the instantaneous velocities of the $k$-th transmit AE and $l$-th receive AE are obtained as the first derivatives of their respective positions:
\begin{align}
	\label{9}
	&{{\bf{v}}_{t}^k\left( \tau  \right)}=\frac{{d{\bf{q}}_{t}^k\left( \tau  \right)}}{{d\tau }},\forall k,\forall \tau,\\
	\label{10}
	&{{\bf{v}}_{r}^l\left( \tau  \right)}=\frac{{d{\bf{q}}_{r}^l\left( \tau  \right)}}{{d\tau }},\forall l,\forall \tau.
\end{align}
Similarly, their corresponding instantaneous accelerations are derived as the second derivatives of positions:
\begin{align}
	\label{11}
	&{{\bf{a}}_{t}^k\left( \tau  \right)}=\frac{{d^2{\bf{q}}_{t}^k\left( \tau  \right)}}{{d\tau^2 }},\forall k,\forall \tau,\\
	\label{12}
	&{{\bf{a}}_{r}^l\left( \tau  \right)}=\frac{{d^2{\bf{q}}_{r}^l\left( \tau  \right)}}{{d\tau^2 }},\forall l,\forall \tau.
\end{align}
Owing to the inherent physical limitations of the actuating motors typically employed in practical CTFA implementations, constraints are imposed on the maximum achievable velocity and acceleration. These limitations are modeled as
\begin{align}
	\label{13}
	&\phantom{s.t.\quad}{\left\|{{\bf{v}}_{t}^k\left( \tau  \right)}\right\|_2}\leq V_{\max},\forall k,\forall \tau,\\
	\label{14}
	&\phantom{s.t.\quad}{\left\|{{\bf{v}}_{r}^l\left( \tau  \right)}\right\|_2}\leq V_{\max},\forall l,\forall \tau,\\
	\label{15}
	&\phantom{s.t.\quad}{\left\|{{\bf{a}}_{t}^k\left( \tau  \right)}\right\|_2}\leq a_{\max},\forall k,\forall \tau,\\
	\label{16}
	&\phantom{s.t.\quad}{\left\|{{\bf{a}}_{r}^l\left( \tau  \right)}\right\|_2}\leq a_{\max},\forall l,\forall \tau,
\end{align}
where $V_{\max}$ and $a_{\max}$ denote the maximum velocity and acceleration magnitudes, respectively.

\subsection{Problem Formulation}
The objective of this paper is to maximize the total throughput $C$ in \eqref{2}, which is achieved by jointly optimizing the AEs' trajectories $\left\{ {{{\tilde {\bf{q}}}_t}\left( \tau \right),{{\tilde {\bf{q}}}_r}\left( \tau \right)\left| {0 \le \tau \le T} \right.} \right\}$, velocities $\left\{ {{{\tilde {\bf{v}}}_t}\left( \tau \right),{{\tilde {\bf{v}}}_r}\left( \tau \right)\left| {0 \le \tau \le T} \right.} \right\}$, accelerations $\left\{ {{{\tilde {\bf{a}}}_t}\left( \tau \right),{{\tilde {\bf{a}}}_r}\left( \tau \right)\left| {0 \le \tau \le T} \right.} \right\}$, and the TCMs $\left\{ {{\bf{Q}}\left( \tau \right)\left| {0 \le \tau \le T} \right.} \right\}$. 
Accordingly, the optimization problem is formulated as
\begin{align}
	\label{17a}
	\left( {{\rm P}1} \right)\quad&\max_{\substack{{\boldsymbol \Omega}_1}}\quad C\tag{17a}\\
	\label{17b}
	&s.t.\quad {{{\bf{q}}}^k_t}\left( \tau \right) \in {{\cal C}_t},\forall k,\forall \tau,\tag{17b}\\
	\label{17c}
	&\phantom{s.t.\quad}{{{{\bf{q}}}^l_r}\left( \tau  \right)\in{\mathcal C}_r},\forall l,\forall \tau,\tag{17c}\\
	\label{17d}
	&\phantom{s.t.\quad}{\left\|\mathbf{q}_t^k\left(\tau\right)-\mathbf{q}_t^{k'}\left(\tau\right)\right\|_2\geq D},k\neq k',\forall \tau,\tag{17d}\\
	\label{17e}
	&\phantom{s.t.\quad}{\left\|\mathbf{q}_r^l\left(\tau\right)-\mathbf{q}_r^{l'}\left(\tau\right)\right\|_2\geq D},l\neq l',\forall \tau,\tag{17e}\\
	\label{17f}
	&\phantom{s.t.\quad}{\rm tr}\left({\mathbf Q}\left(\tau\right)\right)\leq P,\forall \tau,\tag{17f}\\
	\label{17g}
	&\phantom{s.t.\quad}{\mathbf Q}\left(\tau\right)\succeq{
		\mathbf 0},\forall \tau,\tag{17g}\\
	&\phantom{s.t.\quad}\eqref{9}{\rm -}\eqref{16}.\notag
\end{align}
where the set of optimization variables is \begin{small}${\boldsymbol \Omega}_1=\left\{ {{\bf{Q}}\left( \tau \right),{{\tilde {\bf{q}}}_t}\left( \tau \right),{{\tilde {\bf{q}}}_r}\left( \tau \right),} \right.{\tilde {\bf{v}}_t}\left( \tau \right),{\tilde {\bf{v}}_r}\left( \tau \right),{\tilde {\bf{a}}_t}\left( \tau \right), {\tilde {\bf{a}}_r}\left( \tau \right) \left| {0 \le \tau \le T} \right\} $\end{small}, and $D$ in \eqref{17d} and \eqref{17e} signifies the minimum distance between AEs within each time slot to mitigate potential coupling effects.

Directly solving the optimization problem $\left( {{\rm P}1} \right)$ is intrinsically challenging, stemming from four primary impediments. 
First, it necessitates the optimization of continuous variables, namely the AEs' trajectories ${{\tilde {\bf{q}}}}\left( \tau \right)$, along with their associated first- and second-order derivatives (i.e., velocity ${{\tilde {\bf{v}}}}\left( \tau \right)$ and acceleration ${{\tilde {\bf{a}}}}\left( \tau \right)$), in conjunction with the transmit covariance matrix $\mathbf Q\left(\tau\right)$. This inherently constitutes an optimization problem within an infinite-dimensional functional space, rendering direct analytical or numerical solutions intractable. 
Second, the objective function in $\left( {{\rm P}1} \right)$ is formulated as an integral, which lacks a closed-form expression, precluding straightforward evaluation. 
Third, the problem's complexity is further compounded by the nature of the constraints. Specifically, the minimum distance constraints articulated in \eqref{17d} and \eqref{17e} are inherently non-convex. 
Furthermore, the strong coupling among all optimization variables significantly exacerbates the difficulty of attaining a tractable solution for $\left( {{\rm P}1} \right)$.

To enhance the tractability of the optimization problem $\left( {{\rm P}1} \right)$, we utilize a discrete-time linear state-space approximation framework. 
This approach facilitates the derivation of the following expressions, which are predicated upon first- and second-order Taylor series expansions under the stipulation of a sufficiently small time  step $\delta_\tau$:
\setcounter{equation}{17}
\begin{align}
	\label{18}
	& {\tilde{\mathbf v}}_t\left(\tau+\delta_\tau\right) \approx {\tilde{\mathbf v}}_t(\tau)+{\tilde{\mathbf a}}_t(\tau) \delta_\tau, \forall \tau, \\
	\label{19}
	& {\tilde{\mathbf v}}_r\left(\tau+\delta_\tau\right) \approx {\tilde{\mathbf v}}_r(\tau)+{\tilde{\mathbf a}}_r(\tau) \delta_\tau, \forall \tau, \\
	\label{20}
	& {\tilde{\mathbf q}}_t\left(\tau+\delta_\tau\right) \approx {\tilde{\mathbf q}}_t(\tau)+{\tilde{\mathbf v}}_t(\tau) \delta_\tau+\frac{1}{2} {\tilde{\mathbf a}}_t(\tau) \delta_\tau^2, \forall \tau, \\
	\label{21}
	& {\tilde{\mathbf q}}_r\left(\tau+\delta_\tau\right) \approx {\tilde{\mathbf q}}_r(\tau)+{\tilde{\mathbf v}}_r(\tau) \delta_\tau+\frac{1}{2} {\tilde{\mathbf a}}_r(\tau) \delta_\tau^2, \forall \tau.
\end{align}

Consequently, upon discretizing the time duration $T$ into $N+1$ slots, each of duration $\delta_\tau$ (i.e., sampling at time instants $\tau = n\delta_\tau, n=0,1,\cdots,N$), the continuous trajectory ${\tilde{\mathbf q}}_{t/r}(\tau)$ can be accurately characterized by the sequence of discrete-time AEs' positions ${\tilde{\mathbf q}}_{t/r,n}\triangleq{\tilde{\mathbf q}}_{t/r}(n\delta_\tau)$. 
Similarly, the instantaneous velocity and acceleration are represented by their discrete samples ${\tilde{\mathbf v}}_{t/r,n}\triangleq{\tilde{\mathbf v}}_{t/r}(n\delta_\tau)$ and ${\tilde{\mathbf a}}_{t/r,n}\triangleq{\tilde{\mathbf a}}_{t/r}(n\delta_\tau)$, respectively. 
Furthermore, the transmit covariance matrix is discretized as ${{\mathbf Q}}_n\triangleq{{\mathbf Q}}(n\delta_\tau)$. 
This discretization, following the principles of \cite{UAV}, yields the subsequent discrete state-space model as
\begin{align}
	\label{22}
	& {\tilde{\mathbf v}}_{t,n+1} ={\tilde{\mathbf v}}_{t,n}+{\tilde{\mathbf a}}_{t,n} \delta_\tau,\forall n, \\
	\label{23}
	& {\tilde{\mathbf v}}_{r,n+1} ={\tilde{\mathbf v}}_{r,n}+{\tilde{\mathbf a}}_{r,n} \delta_\tau,\forall n, \\
	\label{24}
	& {\tilde{\mathbf q}}_{t,n+1} = {\tilde{\mathbf q}}_{t,n}+{\tilde{\mathbf v}}_{t,n} \delta_\tau+\frac{1}{2} {\tilde{\mathbf a}}_{t,n} \delta_\tau^2,\forall n,\\
	\label{25}
	& {\tilde{\mathbf q}}_{r,n+1} = {\tilde{\mathbf q}}_{r,n}+{\tilde{\mathbf v}}_{r,n} \delta_\tau+\frac{1}{2} {\tilde{\mathbf a}}_{r,n} \delta_\tau^2,\forall n.
\end{align}

More explicitly, the validity of this discrete-time approximation relies on several physical constraints. 
	For the channel to remain approximately constant within each slot, the maximum displacement is required to be substantially smaller than a quarter-wavelength, which implies~$V_{\max}\delta_\tau \ll \lambda/4$. 
	Furthermore, the first-order Taylor approximation of motion holds, provided that the velocity change~$a_{\max}\delta_\tau$ is negligible compared to the maximum velocity, leading to~$a_{\max}\delta_\tau \ll V_{\max}$. 
	Finally, the underlying quasi-static channel assumption mandates that the total duration of~$T=N\delta_\tau$ must be substantially smaller than the channel's coherence time~$T_c$, i.e.,~$N\delta_\tau \ll T_c$. 
	Hence, these constraints collectively dictate that the slot duration~$\delta_\tau$ must be chosen to satisfy~$\delta_\tau \ll \min\left(\frac{T_c}{N}, \frac{\lambda}{4V_{\max}}, \frac{V_{\max}}{a_{\max}}\right)$.

Based upon the preceding mathematical manipulations, Problem (P1) can be reformulated as
\begin{align}
	\label{26a}
	\left( {{\rm P}1^\prime} \right)\quad&\max_{\substack{{\boldsymbol \Omega}_2}}\quad {\hat C}\tag{26a}\\
	\label{26b}
	&s.t.\quad {{ {\bf{q}}}_{t,n}^k}\in {{\cal C}_t},\forall k,\forall n,\tag{26b}\\
	\label{26c}
	&\phantom{s.t.\quad}{{{ {\bf{q}}}_{r,n}^l}\in{\mathcal C}_r},\forall l,\forall n,\tag{26c}\\
	\label{26d}
	&\phantom{s.t.\quad}{\left\|\mathbf{q}_{t,n}^k-\mathbf{q}_{t,n}^{k'}\right\|_2\geq D},k\neq k',\forall n,\tag{26d}\\
	\label{26e}
	&\phantom{s.t.\quad}{\left\|\mathbf{q}_{r,n}^l-\mathbf{q}_{r,n}^{l'}\right\|_2\geq D},l\neq l',\forall n,\tag{26e}\\
	\label{26f}
	&\phantom{s.t.\quad}{\left\|{{\bf{v}}_{t,n}^k}\right\|_2}\leq V_{\max},\forall k,\forall n,\tag{26f}\\
	\label{26g}
	&\phantom{s.t.\quad}{\left\|{{\bf{v}}_{r,n}^l}\right\|_2}\leq V_{\max},\forall l,\forall n,\tag{26g}\\
	\label{26h}
	&\phantom{s.t.\quad}{\left\|{{\bf{a}}_{t,n}^k}\right\|_2}\leq a_{\max},\forall k,\forall n,\tag{26h}\\
	\label{26i}
	&\phantom{s.t.\quad}{\left\|{{\bf{a}}_{r,n}^l}\right\|_2}\leq a_{\max},\forall l,\forall n,\tag{26i}\\
	\label{26j}
	&\phantom{s.t.\quad}{\rm tr}\left({\mathbf Q}_n\right)\leq P,\forall n,\tag{26j}\\
	\label{26k}
	&\phantom{s.t.\quad}{\mathbf Q}_n\succeq{
		\mathbf 0},\forall n,\tag{26k}\\
	&\phantom{s.t.\quad}\eqref{22}{\rm -}\eqref{25}.\notag
\end{align}
where
\setcounter{equation}{26}
\begin{align}
	\label{27}
	&{\boldsymbol \Omega}_2=\left\{ {{{\bf{Q}}_n},{{\tilde {\bf{q}}}_{t,n}},{{\tilde {\bf{q}}}_{r,n}},{{\tilde {\bf{v}}}_{t,n}},{{\tilde {\bf{v}}}_{r,n}},{{\tilde {\bf{a}}}_{t,n}},{{\tilde {\bf{a}}}_{r,n}}} \right\}_{n = 0}^N,\\
	\label{28}
	&\hat C = {\delta _\tau }\sum\limits_{n = 0}^N {{{\log }}\det \left( {{{\bf{I}}} + \frac{1}{{{\sigma ^2}}}{\bf{H}}_n{\bf{Q}}_n{\bf{H}}_n^H} \right)},\\ 
	\label{29} &{\bf{H}}_n=\mathbf{F}(\tilde{\mathbf{q}}_{r,n})^H\mathbf{\Sigma} \mathbf{G}(\tilde{\mathbf{q}}_{t,n}).
\end{align}

\section{Proposed BCD-MM-Based Algorithm}
In this section, we first reformulate $\left( {{\rm P}1^\prime} \right)$ into a more tractable form. Then, the BCD-MM-based algorithm is proposed to alternately optimize the TCM matrices ${\mathbf Q}_n$, the AEs’ trajectories ${\tilde {\bf{q}}_{t/r,n}}$, velocities ${\tilde {\bf{v}}_{t/r,n}}$ and accelerations ${\tilde {\bf{a}}_{t/r,n}}$.

\subsection{Problem Reformulation and Optimization}
Addressing the intractability inherent in optimizing the objective function of $\left( {{\rm P}1'} \right)$, we invoke the well-established weighted minimum mean square error (WMMSE) algorithm \cite{WMMSE1}. 
This technique, commonly utilized for rate and sum-rate maximization problems, is specifically adapted herein to reformulate $\left( {{\rm P}1'} \right)$. 
Fundamentally, the WMMSE approach transforms the original rate/sum-rate maximization objective into an equivalent and alternative formulation via the introduction of pertinent auxiliary variables. 
This reformulation is particularly advantageous as it renders the problem amenable to efficient optimization using BCD methods \cite{WMMSE2}. Such an idea is based on the result in the following lemma.
\setcounter{equation}{12}
\begin{lemma}\label{WMMSE}
	Define an $m$ by $m$ matrix function
	\setcounter{equation}{29}
	\begin{equation}\label{30}
		\mathbb{E}(\mathbf{U}, \mathbf{V}) \triangleq\left(\mathbf{I}-\mathbf{U}^H \mathbf{H V}\right)\left(\mathbf{I}-\mathbf{U}^H \mathbf{H V}\right)^H+\mathbf{U}^H \mathbf{Z} \mathbf{U},
	\end{equation}
	where $\mathbf Z$ is any positive definite matrix. It holds true that
	\begin{equation}\label{31}
		\begin{aligned}
			& \log \operatorname{det}\left(\mathbf{I}+\mathbf{H V} \mathbf{V}^H \mathbf{H}^H \mathbf{Z}^{-1}\right) \\
			= & \max _{\mathbf{W} \succ \mathbf{0}, \mathbf{U}} \log \operatorname{det}(\mathbf{W})-\operatorname{tr}(\mathbf{W}{\mathbb E}(\mathbf{U}, \mathbf{V}))+m.
		\end{aligned}
	\end{equation}
\end{lemma}
The rigorous proof is provided in \cite{WMMSE1}. 
Subsequently, leveraging Lemma~\ref{WMMSE}, an equivalent formulation for the objective function in $\left( {{\rm P}1'} \right)$ is derived via the introduction of pertinent auxiliary variables ${\left\{\mathbf{W}_n,\mathbf{U}_n\right\}_{n=0}^N} $. More specifically, let the matrix $\mathbf{E}_n$ be defined as
\begin{equation}\label{32}
	\mathbf{E}_n \triangleq\left(\mathbf{I}-\frac{1}{\sigma}\mathbf{U}_n^H \mathbf{H}_n \mathbf{Q}_n^{\frac{1}{2}}\right)\left(\mathbf{I}-\frac{1}{\sigma}\mathbf{U}_n^H \mathbf{H}_n \mathbf{Q}_n^{\frac{1}{2}}\right)^H+\mathbf{U}_n^H\mathbf{U}_n,
\end{equation}
where ${\mathbf Q}_n^{\frac{1}{2}}$ represents the matrix square root of ${\mathbf Q_n}$. 
Then, based on the WMMSE framework, the  throughput $\hat C$ admits the following equivalent expression:
\begin{equation}\label{33}
	\begin{split}
		\hat C &= {\delta _\tau }\sum\limits_{n = 0}^N {{{\log }}\det \left( {{{\bf{I}}} + \frac{1}{{{\sigma ^2}}}{\bf{H}}_n{\bf{Q}}_n{\bf{H}}_n^H} \right)} \\
		&=\max _{\left\{\mathbf{W}_n \succ \mathbf{0}, \mathbf{U}_n\right\}_{n=0}^N}  {\delta _\tau }\sum\limits_{n = 0}^N \left[{{{\log }}\det \left( {{\bf{W}}_n} \right)}  - {{\rm{tr}}\left( {{\bf{W}}_n{\bf{E}}_n} \right)} +{N_r}\right].
	\end{split}
\end{equation}
Neglecting the constant term and the overall scaling factor in the objective function, neither of which affects the optimal solution, allows the problem to be reformulated as
\begin{align}
	\label{34a}
	\left( {{\rm P}1''} \right)\quad&\max_{\substack{{\boldsymbol \Omega}_3}}\quad h_1\left({\mathbf \Omega}_3\right)\tag{34a}\\
	\label{34b}
	&s.t.\quad {\mathbf W}_n \succ {\mathbf 0},\tag{34b}\\
	&\phantom{s.t.\quad}\eqref{26b}{\rm -}\eqref{26k},\notag \\
	&\phantom{s.t.\quad}\eqref{22}{\rm -}\eqref{25},\notag
\end{align}
where 
\setcounter{equation}{34}
\begin{align}
	\label{35}
	&{\boldsymbol \Omega}_3=\left\{ {{{\bf{Q}}_n},{{\tilde {\bf{q}}}_{t,n}},{{\tilde {\bf{q}}}_{r,n}},{{\tilde {\bf{v}}}_{t,n}},{{\tilde {\bf{v}}}_{r,n}},{{\tilde {\bf{a}}}_{t,n}},{{\tilde {\bf{a}}}_{r,n}}},{\mathbf W}_n, {\mathbf U}_n \right\}_{n = 0}^N,\\
	\label{36}
	&h_1\left({\mathbf \Omega}_3\right)=\sum\limits_{n = 0}^N \left[{{{\log }}\det \left( {{\bf{W}}_n} \right)}  - {{\rm{tr}}\left( {{\bf{W}}_n{\bf{E}}_n} \right)} \right].
\end{align}

\subsection{Optimize ${\mathbf Q}_n$}
This subsection addresses the optimization of the transmit covariance matrix ${\mathbf Q}_n$. 
Critically, given fixed values for the other optimization variables, problem $\left( {{\rm P}1''} \right)$ reduces to a convex optimization problem with respect to ${\mathbf Q}_n$. 
By virtue of this convexity, the optimal solution relies on the principles of eigenmode transmission. 
Specifically, let the truncated singular value decomposition (SVD) of ${\mathbf H}_n$ be expressed as ${\mathbf H}_n={\mathbf M}_n{\boldsymbol \Xi}_n{\mathbf N}_n^H$, where ${\mathbf M}_n \in {\mathbb C}^{N_r\times S_n}$ and ${\mathbf N}_n \in {\mathbb C}^{N_t\times S_n}$, and ${\boldsymbol \Xi}_n\in{\mathbb C}^{S_n\times S_n}$ is a diagonal matrix containing the $S_n\triangleq{\rm rank}\left({\mathbf H}_n\right)$ non-zero singular values. 
Then, invoking the well-known water-filling solution, the optimal transmit covariance matrix ${\mathbf Q}_n^{\star}$ can be obtained as
\begin{equation}\label{37}
	\boldsymbol{Q}_n^{\star}={{\mathbf N}_n} \operatorname{diag}\left(\left[p_1^{\star}, p_2^{\star}, \ldots, p_{S_n}^{\star}\right]\right) {\mathbf N}_n^H,
\end{equation}
where $p_s^{\star}=\max\left(0,1/\mu-\sigma^2/{{\boldsymbol \Xi}_n}[s,s]^2\right)$, and $\mu$ is the water-level chosen to satisfy $\sum_{s=1}^{S_n}p_s^\star=P$. 

\subsection{Optimize ${\mathbf U}_n$ and ${\mathbf W}_n$}
Observe that the expression in Eq.~\eqref{34a} exhibits concavity with respect to both ${\mathbf W}_n$ and ${\mathbf U}_n$ when the remaining variables are held constant. 
Consequently, this property permits the application of the first-order optimality condition, yielding the optimal solutions for ${\mathbf W}_n$ and ${\mathbf U}_n$ as follows.
\begin{theorem}
The optimal solutions of ${\mathbf W}_n$ and ${\mathbf U}_n$ are given by
\begin{equation}\label{38}
	\begin{aligned}
		\mathbf{U}_n^{\star}&=\arg\max_{\mathbf{U}_n}\quad h_1\left({{\mathbf \Omega}_3}\right)\\&=\frac{1}{\sigma}\left(\mathbf{I}+\frac{1}{\sigma^2}\mathbf{H}_{n}{\mathbf{Q}}_n\mathbf{H}_{n}^{H}\right)^{-1}\mathbf{H}_{n}{\mathbf{Q}}_n^{\frac{1}{2}},
		\end{aligned}
\end{equation}
and
\begin{equation}\label{39}
	\begin{split}
		\mathbf{W}_n^\star
		&=\arg\max_{\mathbf{W}_n\succ{\mathbf 0}}\quad h_1\left({\mathbf \Omega}_3\right)\\
		&=\left[\left(\mathbf{I}-\frac{1}{\sigma}\mathbf{U}_n^{\star H} \mathbf{H}_n \mathbf{Q}_n^{\frac{1}{2}}\right)\left(\mathbf{I}-\frac{1}{\sigma}\mathbf{U}_n^{\star H} \mathbf{H}_n \mathbf{Q}_n^{\frac{1}{2}}\right)^{H}+\mathbf{U}_n^{\star H}\mathbf{U}_n^\star\right]^{-1}.
	\end{split}
\end{equation}
\end{theorem}
\begin{proof}
	See Appendix A.
\end{proof}

\subsection{Optimize $\left\{{{\mathbf q}}_{t,n}^{k},{{\mathbf v}}_{t,n}^{k},{{\mathbf a}}_{t,n}^{k}\right\}_{n=0}^N$}\label{transmit}
The focus of this subsection pertains to optimizing the sequence variables representing the $k$-th transmit AE's position ${{\mathbf q}}_{t,n}^{k}$, velocity ${{\mathbf v}}_{t,n}^{k}$, and acceleration ${{\mathbf a}}_{t,n}^{k}$ across all time indices $n=0, 1, \dots, N$. 
To this end, the pertinent subproblem derived from $\left( {{\rm P}1''} \right)$ is reformulated as
\begin{align}
	\label{40a}
	\left({{\rm P2\text{-}k}}\right)\quad&\min_{\substack{{\mathbf \Omega}_4}}\quad h_2\left({{\mathbf q}}_{t}^{k}\right) \tag{40a}\\
	&s.t.\quad \eqref{26b},\eqref{26d},\eqref{26f},\eqref{26h},\notag\\
	&\phantom{s.t.\quad}{\eqref{22},\eqref{24}}.\notag
\end{align}
where
\setcounter{equation}{40}
\begin{align}
	\label{41}
	&{\mathbf \Omega}_4=\left\{{{\mathbf q}}_{t,n}^{k},{{\mathbf v}}_{t,n}^{k},{{\mathbf a}}_{t,n}^{k}\right\}_{n=0}^N,\\
	\label{42}
	&h_2\left({{\mathbf q}}_{t}^{k}\right)=\sum\limits_{n = 0}^N {{\rm{tr}}\left( {{\bf{W}}_n{\bf{E}}_n} \right)}.
\end{align}
Problem $\left({{\rm P2\text{-}k}}\right)$ is notably intractable, primarily attributed to two factors: 
firstly, the pronounced non-convexity of the objective function with respect to the optimization variables $\left\{{{\mathbf q}}_{t,n}^{k},{{\mathbf v}}_{t,n}^{k},{{\mathbf a}}_{t,n}^{k}\right\}_{n=0}^{N}$, 
and secondly, the coupled relationships among the variables within the constraints, which, particularly due to the minimum distance constraint \eqref{26d}, results in a non-convex feasible region.

To address the challenges delineated above, we first focus on reformulating the objective function of $\left({{\rm P2\text{-}k}}\right)$ into a mathematically more tractable form. 
Owing to the uniform structure exhibited by the summation terms within $h_2\left({{\mathbf q}}_{t}^{k}\right)$, it suffices, without loss of generality, to analyze a single representative term corresponding to the time index $n$ within it. 
Specifically, invoking the channel matrix definitions provided in \eqref{8}, we establish the following theorem.
\begin{theorem}\label{theorem3}
\begin{equation}\label{43}
	\begin{split}
	{{\rm{tr}}\left( {{\bf{W}}_n{\bf{E}}_n} \right)}=&{\mathbf g}\left({\mathbf q}_{t,n}^k\right)^H{\mathbf B}_{t,n}^k{\mathbf g}\left({\mathbf q}_{t,n}^k\right)+2{\rm Re}\left\{{\mathbf g}\left({\mathbf q}_{t,n}^k\right)^H{\mathbf d}_{t,n}^k\right\}\\
	&+{\rm const},
	\end{split}
\end{equation}	
where
	\begin{align}
		\label{44}
		\mathbf Q_n^{\frac{H}{2}} &= \left[ {\begin{array}{*{20}{c}}
				{{\bf{b}}_{X,n}^1}&{{\bf{b}}_{X,n}^2}& \cdots &{{\bf{b}}_{X,n}^{{N_t}}}
		\end{array}} \right], \\
		\label{45}
	     {\mathbf B}_{t,n}^k&\triangleq\left({\mathbf b}_{X,n}^k\right)^H{\mathbf b_{X,n}^k}{\mathbf C}_{t,n}, \\
		\label{46}
		{\mathbf C}_{t,n}&\triangleq\frac{1}{\sigma^2}{\boldsymbol \Sigma}^H{\mathbf F}(\tilde{\mathbf{q}}_{r,n}){\mathbf U}_n{\mathbf W}_n{\mathbf U}_n^H{\mathbf F}(\tilde{\mathbf{q}}_{r,n})^H {\boldsymbol \Sigma},\\
		\label{47}
		{\mathbf d}_{t,n}^k&\triangleq{\mathbf C}_{t,n}\left[\sum_{{i=1},{i\neq k}}^{N_t}\mathbf{g}\left(\mathbf{q}_{t,n}^i\right)\left({\mathbf{b}_{X,n}^{i}}\right)^H\right]{\mathbf{b}_{X,n}^k}-{\boldsymbol{\alpha}}_{t,n}^k,\\
		\label{48}
		{{\bf{A}}_{t,n}} &\triangleq \frac{1}{\sigma }{{\bf{\Sigma }}^H}{\bf{F}}({\tilde {\bf{q}}_{r,n}}){{\bf{U}}_n}{{\bf{W}}_n}{\bf{Q}}_n^{\frac{H}{2}}\\
		\notag
		&=\left[ {\begin{array}{*{20}{c}}
				{{\boldsymbol{\alpha}}_{t,n}^1}&{{\boldsymbol{\alpha}}_{t,n}^2}& \cdots &{{\boldsymbol{\alpha}}_{t,n}^{{N_t}}}
		\end{array}} \right].
	\end{align}
\end{theorem}
\begin{proof}
	See Appendix B.
\end{proof}
Thus, $\left({{\rm P2\text{-}k}}\right)$ can be reformulated as
\begin{align}
	\label{50a}
	\left({{\rm P2'\text{-}k}}
	\right)\quad&\min_{\substack{{\mathbf \Omega}_4}}\quad
	h_3\left({{\mathbf q}}_{t}^{k}\right)\tag{49a}\\
	&s.t.\quad \eqref{26b},\eqref{26d},\eqref{26f},\eqref{26h},\notag\\
	&\phantom{s.t.\quad}{\eqref{22},\eqref{24}}.\notag
\end{align}
where 
\setcounter{equation}{49}
\begin{equation}\label{51}
	\begin{split}
	h_3\left({{\mathbf q}}_{t}^{k}\right)
	&=\sum\limits_{n = 0}^N \underbrace{{\left[ {{\bf{g}}{{({\bf{q}}_{t,n}^k)}^H}{\bf{B}}_{t,n}^k{\bf{g}}({\bf{q}}_{t,n}^k)}+ \right.}\left. { 2{\rm{Re}}\{ {\bf{g}}{{({\bf{q}}_{t,n}^k)}^H}{{\bf{d}}_{t,n}^k}\} } \right]}_{\triangleq {\cal H}_3\left({{ {\bf{q}}}_{t,n}^k}\right)}.
\end{split}
\end{equation}
The non-convexity of problem $\left({{\rm P2'\text{-}k}}\right)$, arising from both its objective function and constraints, necessitates an iterative solution approach. 
Henceforth, the MM algorithm \cite{MM} is invoked to circumvent this intractability. 
Specifically, owing to the structural consistency of the summation terms within the objective function, our analysis focuses, without loss of generality, on the representative constituent term ${\cal H}_3\left({{ {\bf{q}}}_{t,n}^{k}}\right)$. 
Let ${\mathbf q}_{t,n}^{k,(i)}$ denote the solution obtained for the antenna position at the $i$-th iteration, with the corresponding objective function value being $ {\cal H}_3\left({{ {\bf{q}}}_{t,n}^{k,(i)}}\right)$. 
The MM algorithm proceeds by iteratively minimizing a surrogate function. 
Accordingly, at the $(i + 1)$-th iteration, it is required to construct a surrogate function $\gamma\left({{ {\bf{q}}}_{t,n}^{k}}\right)$ that satisfies the three subsequent conditions:
\begin{itemize}
	\item $\gamma\left({{ {\bf{q}}}_{t,n}^{k,(i)}}\right)={\cal H}_3\left({{ {\bf{q}}}_{t,n}^{k,(i)}}\right)$,
	\item $\nabla \gamma\left({{ {\bf{q}}}_{t,n}^{k,(i)}}\right)=\nabla{\cal H}_3\left({{ {\bf{q}}}_{t,n}^{k,(i)}}\right)$,
	\item $\gamma\left({{ {\bf{q}}}_{t,n}^{k}}\right)\geq{\cal H}_3\left({{ {\bf{q}}}_{t,n}^{k}}\right)$.
\end{itemize}

To construct the surrogate function, an upper bound for the first summation term in $ {\cal H}_3\left({{ {\bf{q}}}_{t,n}^{k,(i)}}\right)$ is given as follows.
\begin{lemma}\label{lemma4}
	\begin{equation}\label{52}
		\begin{split}
			&{{\bf{g}}{{({\bf{q}}_{t,n}^k)}^H}{\bf{B}}_{t,n}^k{\bf{g}}({\bf{q}}_{t,n}^k)}\\
			&\leq  {\bf{g}}{{({\bf{q}}_{t,n}^k)}^H}\mathbf{\Phi}_{t,n}^k{\bf{g}}{{({\bf{q}}_{t,n}^k)}}\\
			&\quad-2 \operatorname{Re}\left\{{\bf{g}}{{({\bf{q}}_{t,n}^k)}^H}\left(\mathbf{\Phi}_{t,n}^k-\mathbf{B}_{t,n}^k\right){\bf{g}}{{({\bf{q}}_{t,n}^{k,(i)})}}\right\} \\
			&\quad+{\bf{g}}{{({\bf{q}}_{t,n}^{k,(i)})}^H}\left(\mathbf{\Phi}_{t,n}^k-\mathbf{B}_{t,n}^k\right){\bf{g}}{{({\bf{q}}_{t,n}^{k,(i)})}} \\
			&=\underbrace{L_t\lambda_{\rm max}\left({\mathbf B}_{t,n}^k\right)+{\bf{g}}{{({\bf{q}}_{t,n}^{k,(i)})}^H}\left(\mathbf{\Phi}_{t,n}^k-\mathbf{B}_{t,n}^k\right){\bf{g}}{{({\bf{q}}_{t,n}^{k,(i)})}}}_{\rm const} \\
			&\quad\underbrace{-2 \operatorname{Re}\left\{{\bf{g}}{{({\bf{q}}_{t,n}^k)}^H}\left(\mathbf{\Phi}_{t,n}^k-\mathbf{B}_{t,n}^k\right){\bf{g}}{{({\bf{q}}_{t,n}^{k,(i)})}}\right\}}_{\triangleq \mu_t\left({{{\bf{q}}_{t,n}^k}} \right)},
		\end{split}
	\end{equation}
where ${\boldsymbol {\Phi}}_{t,n}^k=\lambda_{\rm max}\left({\mathbf B}_{t,n}^k\right){\mathbf I}$.
\end{lemma}
\begin{proof}
	Please refer to \cite{LowerBound}.
\end{proof}
According to Lemma \ref{lemma4} and ignoring the constant terms, the surrogate function of $ {\cal H}_3\left({{ {\bf{q}}}_{t,n}^{k,(i)}}\right)$ can be constructed as
\begin{equation}\label{53}
	\begin{split}
	\tau_t\left({{{\bf{q}}_{t,n}^k}}\right) 
	&\triangleq \mu_t\left({{{\bf{q}}_{t,n}^k}}\right)+2 \operatorname{Re}\left\{{\bf{g}}{{({\bf{q}}_{t,n}^k)}^H}{{\bf{d}}_{t,n}^k}\right\}\\
	&=2 \operatorname{Re}\left\{{\bf{g}}{{({\bf{q}}_{t,n}^k)}^H}{{\boldsymbol{\eta}}_{t,n}^{k,(i)}}\right\},
	\end{split}
\end{equation}
where
\begin{equation}\label{54}
	{\boldsymbol \eta}_{t,n}^{k,(i)}={\mathbf d}_{t,n}^k-\left(\mathbf{\Phi}_{t,n}^k-\mathbf{B}_{t,n}^k\right){\bf{g}}{{({\bf{q}}_{t,n}^{k,(i)})}}.
\end{equation} 
Although Eq.~\eqref{53} exhibits linearity with respect to the field response vector ${\bf{g}}{{({\bf{q}}_{t,n}^k)}}$, it nonetheless remains non-convex and non-concave with respect to the AE's position ${\bf{q}}_{t,n}^k$. 
This characteristic precludes straightforward convex optimization techniques. 
Therefore, we resort to employing the second-order Taylor expansion for constructing the surrogate function of $\tau_t\left({{{\bf{q}}_{t,n}^k}}\right)$, which is given by
\begin{equation}\label{55}
	\begin{split}
		&\tau_t\left({{{\bf{q}}_{t,n}^k}}\right)\\
		\leq& \tau_t\left({{{\bf{q}}_{t,n}^{k,(i)}}}\right) + \nabla \tau_t\left({{{\bf{q}}_{t,n}^{k,(i)}}}\right)^T \left( {{\bf{q}}_{t,n}^k - {\bf{q}}_{t,n}^{k,(i)}} \right) \\
		\quad&+ \frac{\delta_{t,n}^k}{2} \left( {{\bf{q}}_{t,n}^k - {\bf{q}}_{t,n}^{k,(i)}} \right)^T \left( {{\bf{q}}_{t,n}^k - {\bf{q}}_{t,n}^{k,(i)}} \right) \\
		=& \underbrace{\frac{\delta_{t,n}^k}{2} \left({{{\bf{q}}_{t,n}^k}}\right)^T {{{\bf{q}}_{t,n}^k}}+ (\nabla \tau_t\left({{{\bf{q}}_{t,n}^{k,(i)}}}\right) - \delta_{t,n}^k {{{\bf{q}}_{t,n}^{k,(i)}}})^T {{{\bf{q}}_{t,n}^k}}}_{\triangleq \gamma_t\left({{ {\bf{q}}}_{t,n}^{k}}\right)} \\
		\quad& + \underbrace{\tau_t\left({{{\bf{q}}_{t,n}^{k,(i)}}}\right) + \frac{\delta_{t,n}^k}{2} \left({{{\bf{q}}_{t,n}^{k,(i)}}}\right)^T {{{\bf{q}}_{t,n}^{k,(i)}}} - \nabla \tau_t\left({{{\bf{q}}_{t,n}^{k,(i)}}}\right)^T {{{\bf{q}}_{t,n}^{k,(i)}}}}_{\rm const},
	\end{split}
\end{equation}
where $\delta_{t,n}^k$ is a positive real number satisfying $\delta_{t,n}^k {\mathbf I}\succeq \nabla^2 \tau_t\left({{{\bf{q}}_{t,n}^k}}\right)$. Following similar derivations in \cite{secure}, the gradient vector $\nabla \tau_t\left({{{\bf{q}}_{t,n}^{k}}}\right)$ and the Hessian matrix $\nabla^2 \tau_t\left({{{\bf{q}}_{t,n}^{k}}}\right)$ with respect to (w.r.t.) ${{{\bf{q}}_{t,n}^{k}}}$ are given in Appendix \ref{Appendix5}, and the construction of $\delta_{t,n}^k$ is presented in Appendix \ref{Appendix6}. Therefore, with the aid of all the surrogate functions $\gamma_t\left({{ {\bf{q}}}_{t,n}^{k}}\right),n=0,1,\cdots,N$, $\left({{\rm P2'\text{-}k}}
\right)$ can be approximated as
\begin{align}
	\label{56a}
	\left({{\rm P2''\text{-}k}}
	\right)\quad&\min_{\substack{{\mathbf \Omega}_4}}\quad
	\sum_{n=0}^{N}\gamma_t\left({{ {\bf{q}}}_{t,n}^{k}}\right)\tag{55a}\\
	&s.t.\quad \eqref{26b},\eqref{26d},\eqref{26f},\eqref{26h},\notag\\
	&\phantom{s.t.\quad}{\eqref{22},\eqref{24}}.\notag
\end{align}
which is still non-convex owing to the constraint \eqref{26d}. Here, we resort to obtaining a sub-optimal local maximum by applying convex relaxation to constraint  \eqref{26d}. More precisely, by conducting the first-order Taylor expansion of the convex function with respect to ${{ {\bf{q}}}_{t,n}^k}$, we can derive a lower bound of $\left\|\mathbf{q}_{t,n}^k-\mathbf{q}_{t,n}^{k'}\right\|_2$ as
\setcounter{equation}{55}
\begin{equation}\label{57}
	\begin{split}
		&\left\|\mathbf{q}_{t,n}^k-\mathbf{q}_{t,n}^{k'}\right\|_2\\ \geq&\left(\nabla\left\|\mathbf{q}_{t,n}^{k,(i)}-\mathbf{q}_{t,n}^{k'}\right\|_2\right)^T\left(\mathbf{q}_{t,n}^k-\mathbf{q}_{t,n}^{k,(i)}\right)+ \left\|\mathbf{q}_{t,n}^{k,(i)}-\mathbf{q}_{t,n}^{k'}\right\|_2\\
	    =&\frac{\left(\mathbf{q}_{t,n}^{k,(i)}-\mathbf{q}_{t,n}^{k'}\right)^T}{\left\|\mathbf{q}_{t,n}^{k,(i)}-\mathbf{q}_{t,n}^{k'}\right\|_2}\left(\mathbf{q}_{t,n}^k-\mathbf{q}_{t,n}^{k'}\right),
	\end{split}
\end{equation}
thereby allowing convex relaxation of the constraint \eqref{26d} to
\begin{equation}\label{58}
	\frac{\left(\mathbf{q}_{t,n}^{k,(i)}-\mathbf{q}_{t,n}^{k'}\right)^T}{\left\|\mathbf{q}_{t,n}^{k,(i)}-\mathbf{q}_{t,n}^{k'}\right\|_2}\left(\mathbf{q}_{t,n}^k-\mathbf{q}_{t,n}^{k'}\right)\geq D,
\end{equation}
and furthermore, the subproblem $\left({{\rm P2''\text{-}k}}
\right)$ can be transformed into
\begin{align}
	\label{59a}
	\left({{\rm P2'''\text{-}k}}
	\right)\quad&\min_{\substack{{{\bf{\Omega}}_4}}}\quad \sum_{n=0}^{N}\gamma_t\left({{ {\bf{q}}}_{t,n}^{k}}\right)\tag{58a}\\
	\label{59b}
	&s.t.\quad\frac{\left(\mathbf{q}_{t,n}^{k,(i)}-\mathbf{q}_{t,n}^{k'}\right)^T}{\left\|\mathbf{q}_{t,n}^{k,(i)}-\mathbf{q}_{t,n}^{k'}\right\|_2}\left(\mathbf{q}_{t,n}^k-\mathbf{q}_{t,n}^{k'}\right) \geq D, k \ne {k'},\tag{58b}\\
	&\phantom{s.t.\quad} \eqref{26b},\eqref{26f},\eqref{26h},\notag\\
	&\phantom{s.t.\quad}{\eqref{22},\eqref{24}},\notag
\end{align}
which is a convex quadratically constrained quadratic programming (QCQP) problem and can be efficiently solved using cvx. The details of the MM algorithm for solving problem (P2-k) are summarized in Algorithm \ref{alg1}.

\begin{algorithm}[htbp]
	\renewcommand{\algorithmicrequire}{\textbf{Input:}}
	\renewcommand{\algorithmicensure}{\textbf{Output:}}
	\caption{MM Algorithm for Solving $\left({{\rm P2\text{-}k}}\right)$}
	\label{alg1}
	\begin{algorithmic}[1]
		\REQUIRE $\bf\Sigma$, $N_t$, $N_r$, $L_r$, $L_t$, $\{\theta_t^p\}_{p=1}^{L_t}$, $\{\phi_t^p\}_{p=1}^{L_t}$, $\{\theta_r^q\}_{q=1}^{L_r}$, $\{\phi_r^q\}_{q=1}^{L_r}$, ${{\cal C}_t}$, ${{\cal C}_r}$, $D$, $\epsilon$.
		\ENSURE  $\left\{{{\mathbf q}}_{t,n}^{k},{{\mathbf v}}_{t,n}^{k},{{\mathbf a}}_{t,n}^{k}\right\}_{n=0}^N$.
		\STATE{Initialization: $i=0$.}
		\STATE{Obtain ${\boldsymbol{\Phi}}_{t,n}^k$ according to Lemma \ref{lemma4}};
		\WHILE{Relative decrease of $h_2\left({{\mathbf q}}_{t,n}^{k}\right)$ is above $\epsilon$}
		\STATE{Obtain ${\boldsymbol \eta}_{t,n}^{k,(i)}$ according to \eqref{54}};
		\STATE{Obtain $\nabla \tau_t\left({{{\bf{q}}_{t,n}^{k}}}\right)$ and $\nabla^2 \tau_t\left({{{\bf{q}}_{t,n}^{k}}}\right)$ according to Appendix \ref{Appendix5};}
		\STATE{Obtain $\delta_{t,n}^k$ according to \eqref{delta};}
		\STATE{Obtain $\left\{{{\mathbf q}}_{t,n}^{k,(i+1)},{{\mathbf v}}_{t,n}^{k,(i+1)},{{\mathbf a}}_{t,n}^{k,(i+1)}\right\}_{n=0}^N$ by solving $\left({{\rm P2'''\text{-}k}}
			\right)$;}
		\STATE{$i=i+1$;}
		\ENDWHILE
	\end{algorithmic}
\end{algorithm}
Next, we analyze the convergence of the proposed Algorithm \ref{alg1}. Denote the two constant terms in \eqref{55} and \eqref{52} by $\Gamma_{t,1}\left({{\mathbf q}}_{t,n}^{k,(i)}\right)$ and $\Gamma_{t,2}\left({{\mathbf q}}_{t,n}^{k,(i)}\right)$. Then, for the $i$-th iteration, the objective function in $\left({{\rm P2'\text{-}k}}
\right)$ can be written as
\setcounter{equation}{58}
\begin{equation}\label{60}
	\begin{split}
		&\sum_{n=0}^{N}{\cal H}_3\left({{\mathbf q}}_{t,n}^{k,(i)}\right)\\
		\stackrel{(a)}{=}& \sum_{n=0}^{N}\gamma_t\left({{ {\bf{q}}}_{t,n}^{k,(i)}}\right)+ \sum_{n=0}^{N}\Gamma_{t,1}\left({{ {\bf{q}}}_{t,n}^{k,(i)}}\right)+ \sum_{n=0}^{N}\Gamma_{t,2}\left({{ {\bf{q}}}_{t,n}^{k,(i)}}\right)\\
		\stackrel{(b)}{\geq}& \sum_{n=0}^{N}\gamma_t\left({{ {\bf{q}}}_{t,n}^{k,(i+1)}}\right)+ \sum_{n=0}^{N}\Gamma_{t,1}\left({{ {\bf{q}}}_{t,n}^{k,(i)}}\right)+ \sum_{n=0}^{N}\Gamma_{t,2}\left({{ {\bf{q}}}_{t,n}^{k,(i)}}\right)\\
		\stackrel{(c)}{\geq}& \sum_{n=0}^{N}\tau_t\left({{ {\bf{q}}}_{t,n}^{k,(i+1)}}\right)+ \sum_{n=0}^{N}\Gamma_{t,2}\left({{ {\bf{q}}}_{t,n}^{k,(i)}}\right)\\
	     \stackrel{(d)}{\geq}& \sum_{n=0}^{N}{\cal H}_3\left({{\mathbf q}}_{t,n}^{k,(i+1)}\right),
	\end{split}
\end{equation}
where the equality marked by (a) holds because the expansion in \eqref{55} and \eqref{52} are tight at ${{\mathbf q}}_{t,n}^{k,(i)}$. The inequality marked by (b) holds because we minimize the value of $\gamma_t\left({{ {\bf{q}}}_{t,n}^{k}}\right)$ in the $i$-th iteration, and the equality can be achieved by choosing ${{\tilde {\bf{q}}}_{t,n}^{(i+1)}}={{\tilde {\bf{q}}}_{t,n}^{(i)}}$. 
The inequalities denoted by (c) and (d) are satisfied owing to the fact that $\gamma_t\left({{ {\bf{q}}}_{t,n}^{k}}\right)+ \Gamma_{t,1}\left({{ {\bf{q}}}_{t,n}^{k,(i)}}\right)$ and $\tau_t\left({{ {\bf{q}}}_{t,n}^{k}}\right)+\Gamma_{t,2}\left({{ {\bf{q}}}_{t,n}^{k,(i)}}\right)$ constitute the constructed surrogate functions for $\tau_t\left({{ {\bf{q}}}_{t,n}^{k}}\right)$ and ${\cal H}_3\left({{ {\bf{q}}}_{t,n}^k}\right)$, respectively. 
A direct consequence of this surrogate function construction is that the sequence of objective function values $\left\{h_3\left({{ {\bf{q}}}_{t}^{k,(i)}}\right)\right\}_{i=0}^{+\infty}$ is guaranteed to be non-increasing, and therefore, it converges to a stationary point, specifically a minimum value in this context.

The computational complexity of the conceived algorithm is detailed as follows. 
Initially, in Step 2, the determination of the maximum eigenvalue of the matrix ${\mathbf B}_{t,n}^k$, typically achieved via eigenvalue decomposition, incurs a complexity on the order of ${\mathcal O}\left(L_t^3\right)$. 
Subsequently, in Step 4, the calculation of ${\boldsymbol \eta}_{t,n}^{k,(i)}$ requires ${\mathcal O}\left(L_t^2\right)$ operations. 
Furthermore, the complexities associated with obtaining the gradient $\nabla \tau_t\left({{{\bf{q}}_{t,n}^{k,(i)}}}\right)$, the Hessian $\nabla^2 \tau_t\left({{{\bf{q}}_{t,n}^{k,(i)}}}\right)$, and the scalar $\delta_{t,k}^n$ in Steps 5 and 6 are ${\mathcal O}\left(L_t\right)$, ${\mathcal O}\left(L_t\right)$, and ${\mathcal O}\left(1\right)$, respectively. 
Finally, in Step 7, solving the convex QCQP subproblem (P5-m) using an interior-point method \cite{QCQP} to attain an accuracy of $\beta$ exhibits a complexity of $O(N^{3.5} \log(1/\beta))$. 
Let $I_t$ denote the maximum number of iterations required for convergence of the loop comprising Steps 4-7. 
Therefore, the overall computational complexity of Algorithm~\ref{alg1} is dominated by ${\mathcal O}\left(L_t^3+I_t\left(L_t^2+N^{3.5} \log(1/\beta)\right)\right)$.

\subsection{Optimize  $\left\{{{\mathbf q}}_{r,n}^{l},{{\mathbf v}}_{r,n}^{l},{{\mathbf a}}_{r,n}^{l}\right\}_{n=0}^N$}
The optimization objective within this subsection is focused on the $l$-th receive AE, specifically determining its position ${{\mathbf q}}_{r,n}^{l}$, velocity ${{\mathbf v}}_{r,n}^{l}$, and acceleration ${{\mathbf a}}_{r,n}^{l}$ for all time indices $n$. 
Consequently, the relevant subproblem is reformulated as
\begin{align}
	\left({{\rm P3\text{-}l}}\right)\quad&\min_{\substack{{\mathbf \Omega}_5}}h_4\left({{\mathbf q}}_{r}^{l}\right) \tag{60a}\\
	&s.t.\quad \eqref{26c},\eqref{26e},\eqref{26g},\eqref{26i},\notag\\
	&\phantom{s.t.\quad}{\eqref{23},\eqref{25}}.\notag
\end{align}
where 
\setcounter{equation}{60}
\begin{align}
	&{\mathbf \Omega}_5=\left\{{{\mathbf q}}_{r,n}^{l},{{\mathbf v}}_{r,n}^{l},{{\mathbf a}}_{r,n}^{l}\right\}_{n=0}^N,\\
	&h_4\left({{\mathbf q}}_{r}^{l}\right)=\sum\limits_{n = 0}^N {{\rm{tr}}\left( {{\bf{W}}_n{\bf{E}}_n} \right)}.
\end{align}
To facilitate the subsequent analysis and render the formulation more tractable, the following theorem, which is analogous to Theorem \ref{theorem3}, is established as follows.
\begin{theorem}

	\begin{equation}
		\begin{split}
			{{\rm{tr}}\left( {{\bf{W}}_n{\bf{E}}_n} \right)}=&{\mathbf f}\left({\mathbf q}_{r,n}^l\right)^H{\mathbf B}_{r,n}^l{\mathbf f}\left({\mathbf q}_{r,n}^l\right)+2{\rm Re}\left\{{\mathbf f}\left({\mathbf q}_{r,n}^l\right)^H{\mathbf d}_{r,n}^l\right\}\\
			&+{\rm const},
		\end{split}
	\end{equation}	
	where 
	\begin{align}
			{{\bf{L}}_{X,n}} &\triangleq {\mathbf U}_n{\mathbf W}_n{\mathbf U}_n^H,\\
			{{\bf{L}}_{X,n}^\frac{H}{2} }&=\left[ {\begin{array}{*{20}{c}}
					{{\bf{l}}_{X,n}^1}&{{\bf{l}}_{X,n}^2}& \cdots &{{\bf{l}}_{X,n}^{{N_r}}}
			\end{array}} \right],\\
			{\mathbf B}_{r,n}^l&\triangleq\left({\mathbf l}_{X,n}^l\right)^H{\mathbf l}_{X,n}^l{\mathbf C}_{r,n},\\
			{\mathbf C}_{r,n}&\triangleq\frac{1}{\sigma^2}{\boldsymbol \Sigma}{\mathbf G}(\tilde{\mathbf{q}}_{t,n}){\mathbf Q}_n{\mathbf G}(\tilde{\mathbf{q}}_{t,n})^H {\boldsymbol \Sigma}^H,\\
			{\mathbf d}_{r,n}^l&\triangleq{\mathbf C}_{r,n}\left[\sum_{{i=1},{i\neq l}}^{N_r}\mathbf{f}\left(\mathbf{q}_{r,n}^i\right)\left(\mathbf{l}_{X,n}^i\right)^H\right]\mathbf{l}_{X,n}^{l}-{\boldsymbol{\alpha}}_{r,n}^l,\\
			{{\bf{A}}_{r,n}} &\triangleq \frac{1}{\sigma }{{\bf{\Sigma }}}{\bf{G}}({\tilde {\bf{q}}_{t,n}}){{\bf{Q}}_n^\frac{1}{2}}{{\bf{W}}_n^H}{\bf{U}}_n^H\\
			\notag
			&=\left[ {\begin{array}{*{20}{c}}
					{{\boldsymbol{\alpha}}_{r,n}^1}&{{\boldsymbol{\alpha}}_{r,n}^2}& \cdots &{{\boldsymbol{\alpha}}_{r,n}^{{N_r}}}
			\end{array}} \right].
	\end{align}
\end{theorem}
\begin{proof}
	See Appendix C.
\end{proof}
Following a similar procedure, by defining $\tau_r\left({{{\bf{q}}_{r,n}^l}}\right) 
\triangleq 2\operatorname{Re}\left\{{\bf{f}}{{({\bf{q}}_{r,n}^l)}^H}{{\boldsymbol{\eta}}_{r,n}^{l,(i)}}\right\}$, where
\begin{equation}\label{72}
	{\boldsymbol \eta}_{r,n}^{l,(i)}\triangleq{\mathbf d}_{r,n}^l-\left(\mathbf{\Phi}_{r,n}^l-\mathbf{B}_{r,n}^l\right){\bf{f}}{{({\bf{q}}_{r,n}^{l,(i)})}}
\end{equation}
with
\begin{equation}\label{73}
	{\boldsymbol {\Phi}}_{r,n}^l=\lambda_{\rm max}\left({\mathbf B}_{r,n}^l\right){\mathbf I},
\end{equation}
the surrogate function can be constructed as $\sum_{n=0}^{N}\gamma_r\left({{ {\bf{q}}}_{r,n}^{l}}\right)$, where $\gamma_r\left({{ {\bf{q}}}_{r,n}^{l}}\right)\triangleq\frac{\delta_{r,n}^l}{2} \left({{{\bf{q}}_{r,n}^l}}\right)^T {{{\bf{q}}_{r,n}^l}}+ (\nabla\tau_r\left( {{{\bf{q}}_{r,n}^{l,(i)}}}\right) - \delta_{r,n}^l {{{\bf{q}}_{r,n}^l}})^T {{{\bf{q}}_{r,n}^l}}$ with $\delta_{r,n}^l {\mathbf I}\succeq \nabla^2 \tau_r\left({{{\bf{q}}_{r,n}^l}}\right)$. In addition, the convex relaxation of the constraint (26e) can be obtained as
\begin{equation}
	\frac{\left(\mathbf{q}_{r,n}^{l,(i)}-\mathbf{q}_{r,n}^{l'}\right)^T}{\left\|\mathbf{q}_{r,n}^{l,(i)}-\mathbf{q}_{r,n}^{l'}\right\|_2}\left(\mathbf{q}_{r,n}^l-\mathbf{q}_{r,n}^{l'}\right) \geq D.
\end{equation}

Finally, the subproblem (P3-l) can be transformed into
\begin{align}
	\left({{\rm P3'\text{-}l}}
	\right)\quad&\min_{\substack{{{\bf{\Omega}}_5}}}\quad \sum_{n=0}^{N}\gamma_r\left({{ {\bf{q}}}_{r,n}^{l}}\right)\tag{73a}\\
	&s.t.\quad\frac{\left(\mathbf{q}_{r,n}^{l,(i)}-\mathbf{q}_{r,n}^{l'}\right)^T}{\left\|\mathbf{q}_{r,n}^{l,(i)}-\mathbf{q}_{r,n}^{l'}\right\|_2}\left(\mathbf{q}_{r,n}^l-\mathbf{q}_{r,n}^{l'}\right) \geq D, l \ne {l'},\tag{73b}\\
	&\phantom{s.t.\quad} \eqref{26c},\eqref{26g},\eqref{26i},\notag\\
	&\phantom{s.t.\quad}{\eqref{23},\eqref{25}}.\notag
\end{align}
which is a convex QCQP problem as well. The details of the MM algorithm for solving problem $\left({{\rm P3\text{-}l}}\right)$  are summarized in Algorithm \ref{alg2}.
\begin{algorithm}[htbp]
	\renewcommand{\algorithmicrequire}{\textbf{Input:}}
	\renewcommand{\algorithmicensure}{\textbf{Output:}}
	\caption{MM Algorithm for Solving Problem $\left({{\rm P3\text{-}l}}\right)$ }
	\label{alg2}
	\begin{algorithmic}[1]
		\REQUIRE $\bf\Sigma$, $N_t$, $N_r$, $L_r$, $L_t$, $\{\theta_t^p\}_{p=1}^{L_t}$, $\{\phi_t^p\}_{p=1}^{L_t}$, $\{\theta_r^q\}_{q=1}^{L_r}$, $\{\phi_r^q\}_{q=1}^{L_r}$, ${{\cal C}_t}$, ${{\cal C}_r}$, $D$, $\epsilon$.
		\ENSURE  $\left\{{{\mathbf q}}_{r,n}^{l},{{\mathbf v}}_{r,n}^{l},{{\mathbf a}}_{r,n}^{l}\right\}_{n=0}^N$.
		\STATE{Initialization: $i=0$.}
		\STATE{Obtain ${\boldsymbol{\Phi}}_{r,n}^l$ according to \eqref{73}};
		\WHILE{Relative decrease of $h_4\left({{\mathbf q}}_{r,n}^{l}\right)$ is above $\epsilon$}
		\STATE{Obtain ${\boldsymbol \eta}_{r,n}^{l,(i)}$ according to \eqref{72}};
		\STATE{Obtain $\nabla \tau_r\left({{{\bf{q}}_{r,n}^{l,(i)}}}\right)$, $\nabla^2 \tau_r\left({{{\bf{q}}_{r,n}^{l,(i)}}}\right)$ and $\delta_{r,n}^l$ as in \cite{secure};}
		\STATE{Obtain $\left\{{{\mathbf q}}_{r,n}^{l,(i+1)},{{\mathbf v}}_{r,n}^{l,(i+1)},{{\mathbf a}}_{r,n}^{l,(i+1)}\right\}_{n=0}^N$ by solving $\left({{\rm P3'\text{-}l}}
			\right)$;}
		\STATE{$i=i+1$;}
		\ENDWHILE
	\end{algorithmic}
\end{algorithm}

Similar to the analysis in Section \ref{transmit}, the monotonic convergence is guaranteed for solving problem $\left({{\rm P3\text{-}l}}\right)$ with Algorithm \ref{alg2}. The corresponding computational complexity is ${\mathcal O}\big(L_r^3+I_r(L_r^2+N^{3.5} \log(1/\beta))\big)$,  with $I_r $ denoting the maximum number of iterations to perform steps 4-6.
\subsection{Overall Algorithm}
Building upon the preceding analysis, the conceived BCD-MM framework for addressing problem $\left( {{\rm P}1^\prime} \right)$ is now completed. 
The overall algorithmic procedure is delineated in Algorithm~\ref{alg3}. 
Specifically, Steps 3 through 7 encompass the sequential optimization of the $N$ TCMs and the two pertinent WMMSE auxiliary variables, according to \eqref{37}, \eqref{38}, and \eqref{39}, respectively. 
Subsequently, during Steps 8 to 10, the positions of the $N_t$ transmit AEs are sequentially optimized by solving problem $\left({{\rm P2'''\text{-}k}}\right)$ leveraging the MM approach. 
Analogously, in Steps 11 to 13, the positions of the $N_r$ receive AEs are iteratively optimized via the solution of subproblem $\left({{\rm P3'\text{-}l}}\right)$, again employing the MM algorithm. 
The algorithm iterates among these constituent subproblem solutions until convergence is attained, which is signified by the incremental increase in the value of $h_1\left({\mathbf \Omega}_3\right)$ falling below a predefined tolerance threshold $\epsilon$. It should be noted that the sequential nature of Algorithm~\ref{alg3} is a characteristic of the BCD-based computational method. In the actual physical implementation, all antenna elements execute their computed optimal trajectories in parallel.
\begin{algorithm}[htbp]
	\renewcommand{\algorithmicrequire}{\textbf{Input:}}
	\renewcommand{\algorithmicensure}{\textbf{Output:}}
	\caption{Alternating Optimization for Solving Problem $\left( {{\rm P}1'} \right)$}
	\label{alg3}
	\begin{algorithmic}[1]
		\REQUIRE $\bf\Sigma$, $N_t$, $N_r$, $L_r$, $L_t$, $\{\theta_t^p\}_{p=1}^{L_t}$, $\{\phi_t^p\}_{p=1}^{L_t}$, $\{\theta_r^q\}_{q=1}^{L_r}$, $\{\phi_r^q\}_{q=1}^{L_r}$, ${{\cal C}_t}$, ${{\cal C}_r}$, $D$, $\epsilon$.
		\ENSURE $\left\{ {{{\bf{Q}}_n},{{\tilde {\bf{q}}}_{t,n}},{{\tilde {\bf{q}}}_{r,n}},{{\tilde {\bf{v}}}_{t,n}},{{\tilde {\bf{v}}}_{r,n}}},{{\tilde {\bf{a}}}_{t,n}},{{\tilde {\bf{a}}}_{r,n}},{\mathbf W}_n, {\mathbf U}_n \right\}_{n = 0}^N$.
		\STATE{Initialization.}
		\WHILE{Relative increase of $h_1\left({\mathbf \Omega}_3\right)$ is above $\epsilon$}
		\FOR{$n=0\rightarrow N$}
		\STATE{Given ${\mathbf \Omega}_3 \setminus \left\{{\mathbf Q}_n\right\}$, update ${\mathbf Q}_n$ via \eqref{37};}
		\STATE{Given ${\mathbf \Omega}_3 \setminus \left\{{\mathbf U}_n\right\}$, update ${\mathbf U}_n$ via \eqref{38};}
		\STATE{Given ${\mathbf \Omega}_3 \setminus \left\{{\mathbf W}_n\right\}$, update ${\mathbf W}_n$ via \eqref{39};}
		\ENDFOR
		\FOR{$k=1\rightarrow N_t$}
		\STATE{Given ${\mathbf \Omega}_3 \setminus \left\{{{\mathbf q}}_{t,n}^{k},{{\mathbf v}}_{t,n}^{k},{{\mathbf a}}_{t,n}^{k}\right\}_{n=0}^N$, solve $\left({{\rm P2\text{-}k}}
			\right)$ to update $\left\{{{\mathbf q}}_{t,n}^{k},{{\mathbf v}}_{t,n}^{k},{{\mathbf a}}_{t,n}^{k}\right\}_{n=0}^N$;}
		\ENDFOR
		\FOR{$l=1\rightarrow N_r$}
	\STATE{Given ${\mathbf \Omega}_3 \setminus \left\{{{\mathbf q}}_{r,n}^{l},{{\mathbf v}}_{r,n}^{l},{{\mathbf a}}_{r,n}^{l}\right\}_{n=0}^N$, solve $\left({{\rm P3\text{-}l}}
		\right)$ to update $\left\{{{\mathbf q}}_{r,n}^{l},{{\mathbf v}}_{r,n}^{l},{{\mathbf a}}_{r,n}^{l}\right\}_{n=0}^N$;}
		\ENDFOR
		\ENDWHILE
	\end{algorithmic}
\end{algorithm}

The convergence characteristics of Algorithm~\ref{alg3} are analyzed as follows. 
By virtue of the BCD approach combined with the MM principle applied to its subproblems, the sequence of objective function values generated for problem $\left( {{\rm P}1'} \right)$ is ensured to be monotonically non-decreasing across consecutive iterations. 
Furthermore, this sequence of objective values is intrinsically upper-bounded due to the finite nature of the underlying throughput. 
Consequently, Algorithm~\ref{alg3} is guaranteed to converge to at least a locally optimal solution for the reformulated problem $\left( {{\rm P}1'} \right)$.

In the following, the computational complexity of Algorithm~\ref{alg3} is analyzed. 
Initially, the computation of the optimal transmit covariance matrices ${\mathbf Q}_n$ via water-filling in Step 4 typically requiring an SVD incurs a complexity of ${\mathcal O}\left( {{N_r}{N_t}\min \left( {{N_r},{N_t}} \right)} \right)$. 
Subsequently, the determination of the auxiliary matrices ${\mathbf U}_n$ and ${\mathbf W}_n$ in Steps 5-6 exhibits complexities of ${\mathcal O}\left(N_r^3\right)$ and ${\mathcal O}\left(N_t^3\right)$, respectively. 
It is pertinent to note that these updates in Steps 4-6 are executed $N$ times within each iteration of the main BCD loop. 
Furthermore, optimizing the transmit FA trajectory variables $\left\{{{\mathbf q}}_{t,n}^{k},{{\mathbf v}}_{t,n}^{k},{{\mathbf a}}_{t,n}^{k}\right\}_{n=0}^{N+1}$ in Steps 8-10 via Algorithm~\ref{alg1} and the receive FA trajectory variables $\left\{{{\mathbf q}}_{r,n}^{k},{{\mathbf v}}_{r,n}^{k},{{\mathbf a}}_{r,n}^{k}\right\}_{n=0}^{N+1}$ in Steps 11-13 via Algorithm~\ref{alg2} are associated with complexities of ${\mathcal O}\left(N_tL_t^3+N_tI_t\left(L_t^2+N^{3.5} \log(1/\beta)\right)\right)$ and ${\mathcal O}\left(N_rL_r^3+N_rI_r\left(L_r^2+N^{3.5} \log(1/\beta)\right)\right)$, respectively, where $I_t$ and $I_r$ denote the maximum inner MM iteration counts. 
Therefore, letting $I_{\rm out}$ be the maximum number of outer BCD iterations, the overall computational complexity of Algorithm~\ref{alg3} is dominated by ${\mathcal O}\left(\left(N{{N_r}{N_t}\min \left( {{N_r},{N_t}} \right)}+NN_t^3+NN_r^3+N_tL_t^3+N_rL_r^3\right.\right.$ $\left.\left.+N_tI_t\left(L_t^2+N^{3.5} \log\frac{1}{\beta}\right)+N_rI_r\left(L_r^2+N^{3.5} \log\frac{1}{\beta}\right)\right)I_{\rm out}\right)$.

\section{Simulation Results}
This section presents simulation results conceived both to validate the efficacy of the proposed algorithm and to demonstrate the throughput advantages afforded by CTFA systems over conventional FPA arrangements. 
In our simulation setup, we stipulate a signal-to-noise ratio (SNR) of 10~dB and employ a carrier frequency of 7.5~GHz, reflecting propagation conditions encountered at higher frequencies. 
To ensure practical relevance, the kinematics of the CTFAs are constrained such that the maximum velocity is limited to $0.016\ \rm{m/s}$ and the maximum acceleration is capped at $0.6\ \rm{m/s^2}$ \cite{motor1,motor2}. Furthermore, the operational area for antenna movement is confined to a square region having a side length of $A=3\lambda$. 
Regarding the multipath channel model, we assume $L_t = L_r = 5$ paths between the transceivers. The path response matrix is modeled as ${\boldsymbol \Sigma} = {\rm diag}\left(\alpha_1, \cdots, \alpha_{L_t}\right)$. Denoting the Rician factor $K$ as the ratio of the average power for line-of-sight (LoS) paths to that for non-line-of-sight (NLoS) paths, the path gains for the NLoS case ($K=0$) follow $\alpha_l \sim \mathcal{CN}(0, 1/L_t)$, while for LoS scenarios ($K > 0$), they follow $\alpha_1 \sim \mathcal{CN}\left(0, \frac{K}{K+1}\right)$ and $\alpha_l \sim \mathcal{CN}\left(0, \frac{1}{(K+1)(L_t-1)}\right)$ for $l \ge 2$.
Additionally, the physical AoDs and AoAs are modeled as i.i.d. random variables, each uniformly distributed over the interval $[0, \pi]$. 
Finally, a minimum inter-AE distance of $D = \lambda/2$ is imposed to mitigate potential mutual coupling effects.
The performance of Algorithm 3 is compared with five benchmark schemes:
\begin{itemize}
	\item [1)] 
	Transmit CTFA (T-CTFA): The receiver is equipped with an FPA-based uniform planar array (UPA), while the transmitter employs CTFAs. In other words, Algorithm~\ref{alg3} is invoked only for the alternating optimization of the remaining variables in ${\mathbf \Omega}_3 \setminus \left\{{\tilde{\mathbf q}}_{r,n},{\tilde{\mathbf v}}_{r,n},{\tilde{\mathbf a}}_{r,n}\right\}_{n=0}^N$.
	\item [2)]
	Linear trajectory I: While maintaining identical original and final antenna positions as the proposed scheme, this baseline trajectory deviates from our total-throughput optimized solution, employing instead a uniform linear motion trajectory.
	\item [3)]
	Linear trajectory II: In this scheme, the final antenna position is derived through instantaneous rate optimization, as detailed in \cite{Capacity}. Furthermore, the antenna trajectory follows a uniform linear motion trajectory.
	\item [4)]
	Random trajectory: In this baseline configuration, both the trajectory and the final antenna position are randomly generated; however, the generated trajectory is constrained to adhere to the stipulations outlined in \eqref{9}-\eqref{16} and \eqref{17b}-\eqref{17g}. 
	\item [5)]
	FPA: This conventional benchmark employs FPA-based UPAs at both the transmitter and receiver. Consequently, within the optimization framework, only the transmit covariance matrices $\left\{{\mathbf Q}_n\right\}_{n=0}^{N}$ are subject to optimization.
\end{itemize}

\begin{figure*}[htbp]
	\centering
	\includegraphics[width=6.5in]{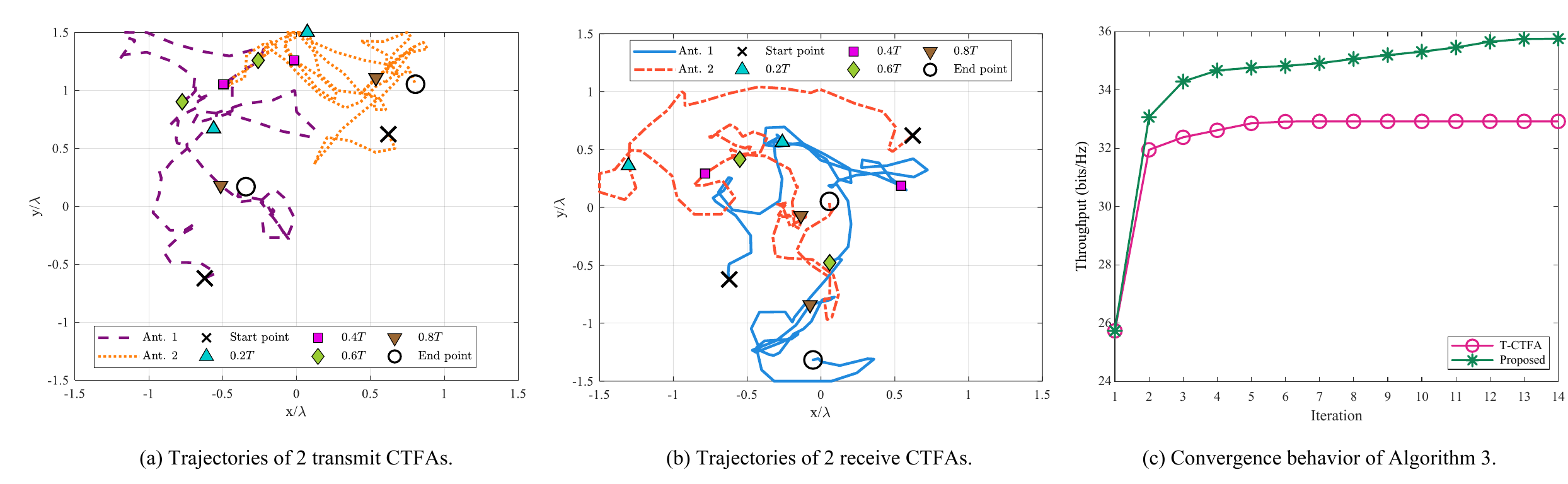}
	\caption{Simulation results for the $2\times2$ CTFA system.}
	\label{fig2}
\end{figure*}
Figs. 2(a), (b), and (c) illustrate the optimized trajectories of the transmit and receive AEs, along with the convergence curve, within the $2\times2$ CTFA system. As depicted in (a) and (b), the antennas exhibit a tendency towards intricate movements to enhance the total throughput. To visualize the optimization process, distinct markers are employed to indicate the antenna positions at different time instants. It can be observed that the antennas continuously adjust their positions from the start point ($\times$) to the end point ($\circ$) to dynamically adapt to the channel environment. This behavior is attributed, in part, to the constraints imposed by the minimum antenna separation and kinematic limitations, which restrict the DoFs in antenna motion. Furthermore, these complex trajectories are also indicative of the optimization process's endeavor to circumvent local optima, thereby unlocking greater performance potential.
Fig. 2(c) illustrates the convergence performance of Algorithm 3, revealing that both the transmit CTFA scheme and the proposed scheme exhibit rapid convergence to their respective optimal values within 14 iterations. This rapid convergence underscores the efficiency of our proposed algorithm. Quantitatively, in comparison to their respective initial values, the optimized solutions achieved by the two schemes demonstrate substantial performance enhancements of $38.6\%$ and $15.9\%$, respectively.

\begin{figure*}[htbp]
	\centering
	\includegraphics[width=6.5in]{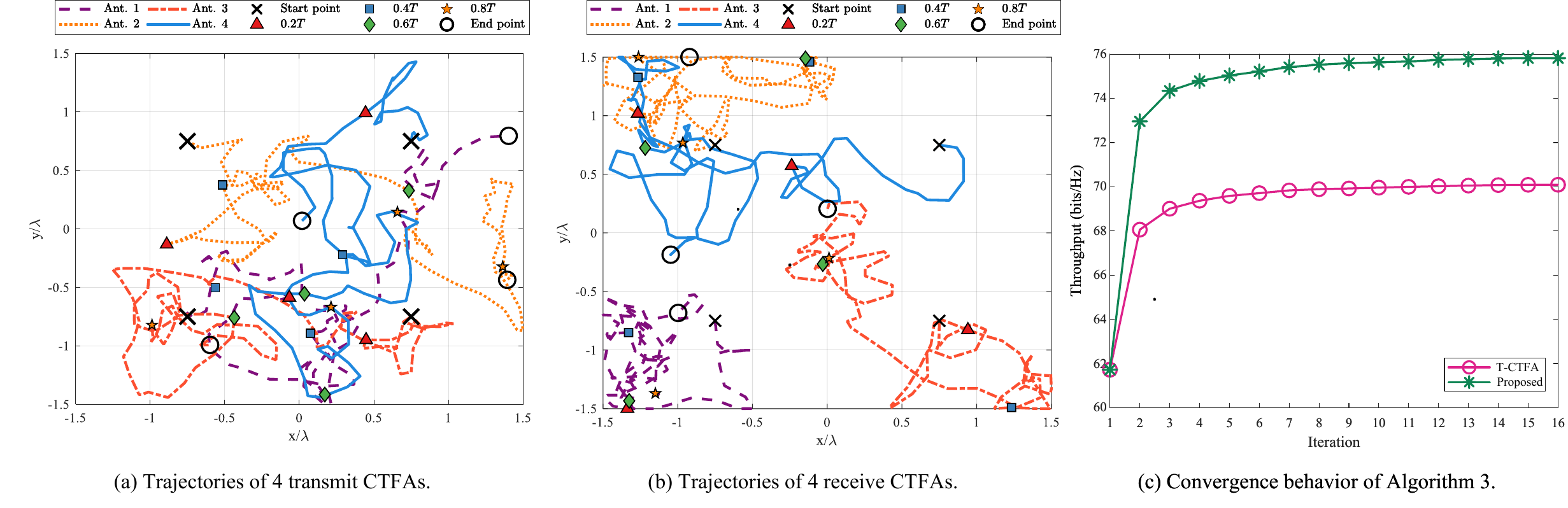}
	\caption{Simulation results for the $4\times4$ CTFA system.}
	\label{fig3}
\end{figure*}

Figs. 3(a), (b), and (c) present simulation results for the $4\times4$ CTFA system. In contrast to the results in Fig. 2, Fig. 3(a) and (b) reveal more intricate antenna movement patterns. This increased complexity arises from the constrained movement area and the heightened probability of antenna collisions, necessitating more sophisticated trajectory planning to achieve a superior compromise between physical limitations and performance. It is crucial to emphasize, however, that this challenge can be readily addressed by expanding the movement area. At higher carrier frequencies, this trade-off becomes negligible, as achieving ideal performance typically requires an area side length commensurate with the wavelength scale. Fig. 3(c) illustrates the convergence performance of Algorithm 3 within this $4\times4$ system, exhibiting a rapid convergence characteristic akin to that observed in Fig. 2(c). This consistent convergence behavior suggests that our algorithm is not unduly sensitive to the increase in the number of antennas, highlighting its potential for low-complexity operation even in larger-scale antenna configurations. Within the $4\times4$ system, both the transmit CTFA and the proposed schemes demonstrate performance improvements of approximately $11.6\%$ and $22.6\%$, respectively, relative to their initial values. This marginal reduction in performance gains is attributed to the further constrained DoFs in antenna movement.

\begin{table}[htbp]
	\centering
	\caption{Comparison of Total Throughput and Percentage Gain over FPA for Different Schemes.}
	\label{tab:my_comparison}
	\renewcommand\arraystretch{1.3}
	\resizebox{\textwidth}{!}{
	\begin{tabular}{|c|c|c|c|c|c|c|c|}
		\hline
		Configuration & Metric & Proposed & Transmit CTFA & Linear trajectory I & Linear trajectory II & Random trajectory & FPA \\ \hline
		\multirow{2}{*}{$2\times2$}
		& Throughput (bits/Hz) & $35.76$ & $32.92$ & $28.79$ & $27.48$ & $25.9$ & $27.06$ \\ \cline{2-8}
		& Gain over FPA & $32.15\%$ & $21.66\%$ & $6.39\%$ & $1.55\%$ & $-4.29\%$ & $-$\\ \hhline{|=|=|=|=|=|=|=|=|}
		\multirow{2}{*}{$4\times4$}
		& Throughput (bits/Hz) & $75.82$ & $70.09$ & $67.10$ & $66.17$ & $55.18$ & $58.40$ \\ \cline{2-8}
		& Gain over FPA & $29.83\%$ & $19.98\%$ & $14.89\%$ & $13.30\%$ & $-5.51\%$ & $-$ \\ \hline
	\end{tabular}
}
\end{table}

\begin{table*}[!t]
	\centering
	\caption{Sensitivity Analysis of Convergence and Throughput under Different Channel Conditions.}
	\label{tab:sensitivity_analysis}
	\renewcommand\arraystretch{1.3}
	\begin{tabular}{|c|c|c|c|c|c|}
		\hline
		\multicolumn{2}{|c|}{\textbf{Fixed} $L=5$} & NLoS, $K=0$ & LoS, $K=1$ & LoS, $K=5$ & LoS, $K=10$ \\ \hline
		\multirow{2}{*}{$2\times2$}
		& Iterations & $14$ & $13$ & $16$ & $11$ \\ \cline{2-6}
		& Throughput (bits/Hz) & $35.76$ & $30.23$ & $25.93$ & $24.27$ \\ \hhline{|=|=|=|=|=|=|}
		\multirow{2}{*}{$4\times4$}
		& Iterations & $16$ & $46$ & $36$ & $10$ \\ \cline{2-6}
		& Throughput (bits/Hz) & $75.82$ & $67.18$ & $55.32$ & $48.92$ \\ \hline
		
		\multicolumn{6}{c}{} \\ [-1em]
		\hline
		
		\multicolumn{2}{|c|}{\textbf{Fixed NLoS}} & $L=5$ & $L=10$ & $L=15$ & $L=20$ \\ \hline
		\multirow{2}{*}{$2\times2$}
		& Iterations & $14$ & $5$ & $5$ & $7$ \\ \cline{2-6}
		& Throughput (bits/Hz) & $35.76$ & $34.84$ & $29.77$ & $37.49$ \\ \hhline{|=|=|=|=|=|=|}
		\multirow{2}{*}{$4\times4$}
		& Iterations & $16$ & $24$ & $16$ & $29$ \\ \cline{2-6}
		& Throughput (bits/Hz) & $75.82$ & $56.53$ & $68.02$ & $62.31$ \\ \hline
	\end{tabular}
\end{table*}

Table~I presents a comparative analysis of the total throughput attained by the various schemes. Several salient trends warrant highlighting:
Firstly, the results confirm the universal performance enhancement of CTFA technology. All evaluated CTFA-based schemes consistently outperform the conventional FPA benchmark, unequivocally validating the intrinsic advantage of leveraging antenna mobility to mitigate fading via additional spatial DoFs.
Secondly, the paramount importance of joint trajectory and transmit covariance optimization is evident. Our proposed scheme achieves the highest throughput in both $2\times2$ and $4\times4$ configurations, representing substantial gains of 32.15\% and 29.83\% over the FPA baseline. These gains substantiate the efficacy of our joint optimization strategy. Notably, the proposed scheme significantly surpasses benchmarks relying on simplified uniform linear motion, indicating that meticulous optimization of the entire trajectory is crucial for maximizing CTFA performance.
Thirdly, the MIMO configuration influences the gain magnitude, yet CTFA advantages remain significant. Although the relative gain is marginally lower in the $4\times4$ scenario (29.83\%) versus $2\times2$ case (32.15\%), attributable to tighter constraints from more antennas limiting positioning freedom, the absolute performance improvement remains considerable, corroborating CTFA's value in complex multi-antenna systems.
Furthermore, linear trajectory I consistently outperforms linear trajectory II, which suggests that leveraging even partial trajectory information, i.e., the optimized endpoints, yields better results than optimizing solely for the final position, reinforcing the importance of the trajectory, although both linear-trajectory schemes are suboptimal compared to our full trajectory optimization.
In summary, the results validate our optimization strategy and fundamentally demonstrate the necessity and substantial benefits of meticulous trajectory optimization in CTFA systems.

Finally, to quantify the robustness of the proposed algorithm against channel variations, Table~\ref{tab:sensitivity_analysis} characterizes the convergence iterations and the attainable total throughput associated with varying Rician $K$-factors and numbers of propagation paths ($L$). Two key observations may be drawn from Table~\ref{tab:sensitivity_analysis}. First, upon increasing the Rician $K$-factor from $0$ (NLoS) to $10$, the total throughput is naturally eroded due to the paucity of spatial DoFs in LoS-dominated environments. Nonetheless, the algorithm exhibits a robust convergence behavior; more explicitly, in the $4 \times 4$ scenario, the number of iterations required is significantly reduced at $K=10$, which is attributed to the more benign optimization landscape. Second, varying the number of multipath clusters $L$ under NLoS conditions demonstrates that the proposed algorithm is capable of sustaining a high throughput as well as a stable convergence efficiency, regardless of the scattering richness. These results validate the efficacy of the proposed scheme across diverse propagation environments.

\section{Conclusion}
In this paper, we addressed the limitations inherent in idealized switching models commonly adopted for FA systems. 
Specifically, by introducing the concept of CTFA, we explicitly incorporated the continuous physical movement and realistic kinematic constraints pertinent to practical motor-driven implementations. 
Within this CTFA framework, the problem of maximizing total throughput was formulated, which necessitates the joint optimization of the continuous AEs' trajectories and the associated transmit covariance matrices, subject to kinematic constraints. 
To circumvent the intractability of the resulting non-convex joint optimization problem, an efficient iterative algorithm predicated upon BCD and MM principles was conceived. 
Our simulation results substantiated the efficacy of the proposed CTFA framework and the conceived optimization algorithm, demonstrating significant total throughput enhancements relative to conventional FPA systems and alternative benchmarks employing simplified trajectory assumptions. In our future work, we will extend the proposed optimization framework to the challenging near-field regime.

\appendices
\section{Proof of Theorem 2}
When optimizing ${\mathbf U}_n$, the problem can be formulated as
\setcounter{equation}{73}
\begin{equation}
	\mathbf{U}_n^{\star}
	=\arg\min_{\mathbf{U}_n}\quad {{\rm{tr}}\left( {{\bf{W}}_n{\bf{E}}_n} \right)},
\end{equation}
and with \eqref{32}, the objective function can be further expanded as
\begin{equation}
	\begin{split}
		&{\rm{tr}}({{\bf{W}}_n}{{\bf{E}}_n}) \\
		=& {\rm{tr}}({{\bf{W}}_n}) - \frac{1}{\sigma }{\rm{tr}}({{\bf{W}}_n}({\bf{U}}_n^H{{\bf{H}}_n}{\bf{Q}}_n^{\frac{1}{2}} + {\bf{Q}}_n^{\frac{H}{2}}{\bf{H}}_n^H{{\bf{U}}_n}))\\
		\quad&+ \frac{1}{{{\sigma ^2}}}{\rm{tr}}({{\bf{W}}_n}{\bf{U}}_n^H{{\bf{H}}_n}{{\bf{Q}}_n}{\bf{H}}_n^H{{\bf{U}}_n}) + {\rm{tr}}({{\bf{W}}_n}{\bf{U}}_n^H{{\bf{U}}_n}).
	\end{split}
\end{equation}
Ignore the constant terms and define a new objective function as
\begin{equation}
	\begin{split}
		&\Phi \left( {{{\bf{U}}_n}} \right) \\
		\triangleq&  - \frac{1}{\sigma }{\rm{tr}}({{\bf{W}}_n}({\bf{U}}_n^H{{\bf{H}}_n}{\bf{Q}}_n^{\frac{1}{2}} + {\bf{Q}}_n^{\frac{H}{2}}{\bf{H}}_n^H{{\bf{U}}_n}))\\
		\quad&+ \frac{1}{{{\sigma ^2}}}{\rm{tr}}({{\bf{W}}_n}{\bf{U}}_n^H{{\bf{H}}_n}{{\bf{Q}}_n}{\bf{H}}_n^H{{\bf{U}}_n}) + {\rm{tr}}({{\bf{W}}_n}{\bf{U}}_n^H{{\bf{U}}_n}),
	\end{split}
\end{equation}
whose conjugate gradient with respect to ${{{\bf{U}}_n}}$ can be given as
\begin{equation}
	\begin{split}
		\nabla \Phi \left( {{{\bf{U}}_n}} \right) &=  - \frac{1}{\sigma }{{\bf{H}}_n}{\bf{Q}}_n^{\frac{1}{2}}{{\bf{W}}_n} + \frac{1}{{{\sigma ^2}}}{{\bf{H}}_n}{{\bf{Q}}_n}{\bf{H}}_n^H{{\bf{U}}_n}{{\bf{W}}_n} \\
		&\quad+ {{\bf{U}}_n}{{\bf{W}}_n}.
	\end{split}
\end{equation}
Solving $	\nabla \Phi \left( {{{\bf{U}}_n}} \right) ={\mathbf 0}$ leads to
\begin{equation}
	\mathbf{U}_n^{\star}=\frac{1}{\sigma}\left(\mathbf{I}+\frac{1}{\sigma^2}\mathbf{H}_{n}{\mathbf{Q}}_n\mathbf{H}_{n}^{H}\right)^{-1}\mathbf{H}_{n}{\mathbf{Q}}_n^{\frac{1}{2}}.
\end{equation}

When optimizing ${\mathbf W}_n$, the problem can be formulated as
\begin{equation}
	\mathbf{W}_n^{\star}
	=\arg\max_{\mathbf{W}_n \succ \mathbf{0}}\quad \log \det \left( {{{\bf{W}}_n}} \right) - {\rm{tr}}\left( {{{\bf{W}}_n}{{\bf{E}}_n}} \right),
\end{equation}
whose Lagrangian is
\begin{equation}
	L(\mathbf{W}_n, \mathbf{\Lambda}) = \log \det (\mathbf{W}_n) - \operatorname{tr} (\mathbf{W}_n \mathbf{E}_n)+ \operatorname{tr} (\mathbf{\Lambda} \mathbf{W}_n).
\end{equation}
The gradient of Lagrangian is
\begin{equation}
	\nabla L({{\bf{W}}_n},{\bf{\Lambda }}) = {\bf{W}}_n^{ - T} - {\bf{E}}_n^T + {{\bf{\Lambda }}^T},
\end{equation}
and thus, the KKT conditions are
\begin{align}
		&{\bf{W}}_n^{ - 1} - {{\bf{E}}_n} + {\bf{\Lambda }} = {\bf{0}},\label{83}\\
		&{\bf{\Lambda }}{{\bf{W}}_n}= {\mathbf 0},\label{84}\\
		&{{\bf{W}}_n} \succ {\bf{0}},{\bf{\Lambda }} \succeq {\bf{0}}.
\end{align}
Eq. \eqref{83} implies that
\begin{equation}
	{{\bf{W}}_n} = {\left( {{{\bf{E}}_n} - {\bf{\Lambda }}} \right)^{ - 1}}.
\end{equation}
Since ${{\bf{W}}_n} \succ {\bf{0}}$, it is trivial to derive from \eqref{84} that ${\mathbf \Lambda}={\mathbf 0}$, and we finally arrive at
\begin{equation}
		\begin{split}
			\mathbf{W}_n^\star
			&= {{\bf{E}}_n^{ - 1}}\\
			&=\left[\left(\mathbf{I}-\frac{1}{\sigma}\mathbf{U}_n^{\star H} \mathbf{H}_n \mathbf{Q}_n^{\frac{1}{2}}\right)\left(\mathbf{I}-\frac{1}{\sigma}\mathbf{U}_n^{\star H} \mathbf{H}_n \mathbf{Q}_n^{\frac{1}{2}}\right)^{H}+\mathbf{U}_n^{\star H}\mathbf{U}_n^\star\right]^{-1},
		\end{split}
\end{equation}
which completes the proof.

\section{Proof of Theorem 3}
\label{Theorem2}
Leveraging \eqref{32}, we have
\begin{equation}\label{89}
		{\rm{tr}}\left( {{{\bf{W}}_n}{{\bf{E}}_n}} \right)= f_{t,1} + f_{t,2} + \underbrace {{\rm{tr}}\left( {{{\bf{W}}_n} + {{\bf{W}}_n}{\bf{U}}_n^H{{\bf{U}}_n}} \right)}_{\rm const},
\end{equation}	
where
\begin{equation}\label{90}
	\begin{split}
		{f_{t,1}} &=  - 2{\mathop{\rm Re}\nolimits} \left\{ {{\rm{tr}}\left( {\frac{1}{\sigma }{{\bf{W}}_n}{\bf{Q}}_n^{\frac{H}{2}}{\bf{G}}_n^H{
					\boldsymbol \Sigma ^H}{{\bf{F}}_n}{{\bf{U}}_n}} \right)} \right\}\\
		&=  - 2{\mathop{\rm Re}\nolimits} \left\{ {{\rm{tr}}\left( {\underbrace {\frac{1}{\sigma }{\boldsymbol \Sigma ^H}{{\bf{F}}_n}{{\bf{U}}_n}{{\bf{W}}_n}{\bf{Q}}_n^{\frac{H}{2}}}_{ \buildrel \Delta \over = {\bf{A}}_{t,n} = \left[ {\begin{array}{*{20}{c}}
							{{\boldsymbol{\alpha}}_{t,n}^1}&{{\boldsymbol{\alpha}}_{t,n}^2}& \cdots &{{\boldsymbol{\alpha}}_{t,n}^{{N_t}}}
					\end{array}} \right]}{\bf{G}}_n^H} \right)} \right\}\\
		&=  - 2{\rm{Re}}\left\{ {{\bf{g}}{{\left( {{\bf{q}}_{t,n}^k} \right)}^H}{{\boldsymbol{\alpha}}}_{t,n}^k} \right\}  \underbrace {-2{\rm{Re}}\left\{ {\sum\limits_{i= 1,\hfill\atop
					i \ne k\hfill}^{{N_t}} {{\bf{g}}{{\left( {{\bf{q}}_{t,n}^i} \right)}^H}{{\boldsymbol{\alpha}}}_{t,n}^i} } \right\}}_{\rm const},
	\end{split}
\end{equation}
and
\begin{equation}\label{91}
	\begin{split}
		{f_{t,2}} &= {\rm{tr}}\left( {\frac{1}{{{\sigma ^2}}}{{\bf{W}}_n}{\bf{U}}_n^H{\bf{F}}_n^H{\boldsymbol \Sigma} {{\bf{G}}_n}{{\bf{Q}}_n}{\bf{G}}_n^H{{\boldsymbol \Sigma} ^H}{{\bf{F}}_n}{{\bf{U}}_n}} \right)\\
		&= {\rm{tr}}\left( {{{\bf{G}}_n}{{\bf{Q}}_n}{\bf{G}}_n^H\underbrace {\frac{1}{{{\sigma ^2}}}{{\boldsymbol \Sigma} ^H}{{\bf{F}}_n}{{\bf{U}}_n}{{\bf{W}}_n}{\bf{U}}_n^H{\bf{F}}_n^H{\boldsymbol \Sigma} }_{{{\bf{C}}_{t,n}}}} \right)\\
		&= {\rm{tr}}\left[ {\left( {\sum\limits_{i = 1}^{{N_t}} {{\bf{g}}\left( {{\bf{q}}_{t,n}^i} \right)\left({\bf{b}_{X,n}^i}\right)^H} } \right){{\left( {\sum\limits_{i = 1}^{{N_t}} {{\bf{g}}\left( {{\bf{q}}_{t,n}^i} \right)\left({\bf{b}_{X,n}^i}\right)^H} } \right)}^H}{{\bf{C}}_{t,n}}} \right]\\
		&= {\bf{g}}{\left( {{\bf{q}}_{t,n}^k} \right)^H}\underbrace {{\left[\left({\bf{b}}_{X,n}^k\right)^H{{\bf{b}}_{X,n}^k}\right]}{{\bf{C}}_{t,n}}}_{{\bf{B}}_{t,n}^k}{\bf{g}}\left( {{\bf{q}}_{t,n}^k} \right) \\
		&\quad+ {2{\rm{Re}}\left\{ {{\bf{g}}{{\left( {{\bf{q}}_{t,n}^k} \right)}^H}{{\bf{C}}_{t,n}}\left( {\sum\limits_{i = 1,\hfill\atop
						i \ne k\hfill}^{{N_t}} {{\bf{g}}\left( {{\bf{q}}_{t,n}^i} \right)\left({\bf{b}_{X,n}^i}\right)^H} } \right){{\bf{b}}_{X,n}^k}} \right\}}\\
		&\quad+ \underbrace {\left[ {\left( {\sum\limits_{i= 1,\hfill\atop
						i \ne k\hfill}^{{N_t}} {{\bf{g}}\left( {{\bf{q}}_{t,n}^i} \right)\left({\bf{b}_{X,n}^i}\right)^H} } \right){{\left( {\sum\limits_{i= 1,\hfill\atop
								i \ne k\hfill}^{{N_t}} {{\bf{g}}\left( {{\bf{q}}_{t,n}^i} \right)\left({\bf{b}_{X,n}^i}\right)^H} } \right)}^H}{{\bf{C}}_{t,n}}} \right]}_{\rm const}.
	\end{split}
\end{equation}
Substituting \eqref{90} and \eqref{91} into \eqref{89} completes the proof.

\section{Proof of Theorem 4}
\label{Theorem4}
Similar to Appendix B, we have
\begin{equation}\label{92}
		{\rm{tr}}\left( {{{\bf{W}}_n}{{\bf{E}}_n}} \right)= f_{r,1}+f_{r,2}+ \underbrace {{\rm{tr}}\left( {{{\bf{W}}_n} + {{\bf{W}}_n}{\bf{U}}_n^H{{\bf{U}}_n}} \right)}_{{\rm{const}}},
\end{equation}
where
\begin{equation}\label{93}
	\begin{split}
		{f_{r,1}} &=  - 2{\mathop{\rm Re}\nolimits} \left\{ {{\rm{tr}}\left( {\underbrace {\frac{1}{\sigma }{\bf{\Sigma }}{{\bf{G}}_n}{\bf{Q}}_n^{\frac{1}{2}}{\bf{W}}_n^H{\bf{U}}_n^H}_{ \buildrel \Delta \over = {\bf{A}}_n^r = \left[ {\begin{array}{*{20}{c}}
							{{\boldsymbol{\alpha}}_{r,n}^1}&{{\boldsymbol{\alpha}}_{r,n}^2}& \cdots &{{\boldsymbol{\alpha}}_{r,n}^{{N_t}}}
					\end{array}} \right]}{\bf{F}}_n^H} \right)} \right\}\\
		&=  - 2{\rm{Re}}\left\{ {{\bf{f}}{{\left( {{\bf{q}}_{r,n}^l} \right)}^H}{{\boldsymbol{\alpha}}}_{r,n}^l} \right\} - \underbrace {2{\rm{Re}}\left\{ {\sum\limits_{\scriptstyle i = 1,\hfill\atop
					\scriptstyle i \ne l\hfill}^{{N_t}} {{\bf{f}}{{\left( {{\bf{q}}_{r,n}^i} \right)}^H}{{\boldsymbol{\alpha}}}_{r,n}^i} } \right\}}_{{\rm{const}}},
	\end{split}
\end{equation}
and
\begin{equation}\label{94}
	\begin{split}
		{f_{r,2}}
		&= {\bf{f}}{\left( {{\bf{q}}_{r,n}^l} \right)^H}\underbrace {\left({\bf{l}}_{X,n}^l\right)^H{{\bf{l}}_{X,n}^l}{{\bf{C}}_{r,n}}}_{{\bf{B}}_{r,n}^l}{\bf{f}}\left( {{\bf{q}}_{r,n}^l} \right) \\
		&\quad+ 2{\rm{Re}}\left\{ {{\bf{f}}{{\left( {{\bf{q}}_{r,n}^l} \right)}^H}{{\bf{C}}_{r,n}}\left( {\sum\limits_{\scriptstyle i  = 1,\hfill\atop
						\scriptstyle i \ne l\hfill}^{{N_t}} {{\bf{f}}\left( {{\bf{q}}_{r,n}^i} \right)\left({\bf{l}}_{X,n}^i\right)^H} } \right){{\bf{l}}_{X,n}^l}} \right\}\\
		&\quad +\underbrace{\left[ {\left( {\sum\limits_{i= 1,\hfill\atop
						i \ne l\hfill}^{{N_t}} {{\bf{f}}\left( {{\bf{q}}_{r,n}^i} \right)\left({\bf{l}}_{X,n}^i\right)^H} } \right){{\left( {\sum\limits_{i= 1,\hfill\atop
								i \ne l\hfill}^{{N_t}} {{\bf{f}}\left( {{\bf{q}}_{r,n}^l} \right)\left({\bf{l}}_{X,n}^i\right)^H} } \right)}^H}{{\bf{C}}_{r,n}}} \right]}_{\rm const}.
	\end{split}
\end{equation}
Substituting \eqref{93} and \eqref{94} into \eqref{92} completes the proof.

\section{Calculation of $\nabla \tau_t\left({{{\bf{q}}_{t,n}^{k}}}\right)$ and $\nabla^2 \tau_t\left({{{\bf{q}}_{t,n}^{k}}}\right)$}
\label{Appendix5}
Denoting the $i$-th element of ${\boldsymbol \eta}_{t,n}^{k,(i)}$ as $\eta_{i}$, $\tau_t\left({{{\bf{q}}_{t,n}^{k}}}\right)$ can be written as
\begin{equation}
	\begin{split}
			{\tau _t}\left( {{\bf{q}}_{t,n}^k} \right) &= 2{\mathop{\rm Re}\nolimits} \left\{ {{\bf{g}}{{({\bf{q}}_{t,n}^k)}^H}{\bf{\eta }}_{t,n}^{k,(i)}} \right\}\\
			&= {\mathop{\rm Re}\nolimits} \left\{ {2\sum\limits_{i = 1}^{{L_t}} {{\eta _i}{e^{ - j\frac{{2\pi }}{\lambda }\left( {x_{t,n}^k\sin \theta _t^i\cos \phi _t^i + y_{t,n}^k\cos \theta _t^i} \right)}}} } \right\}\\
			&= 2\sum\limits_{i = 1}^{{L_t}} {\left| {{\eta _i}} \right|\cos \left( {{\kappa _i}\left( {{\bf{q}}_{t,n}^k} \right)} \right)} 
	\end{split}
\end{equation}
where ${\kappa _i}\left( {{\bf{q}}_{t,n}^k} \right) \triangleq \frac{{2\pi }}{\lambda }\left( {x_{t,n}^k\sin \theta _t^i\cos \phi _t^i + y_{t,n}^k\cos \theta _t^i} \right) - \angle {\eta _i}$. Then, the gradient vector and Hessian matrix of ${\tau _t}\left( {{\bf{q}}_{t,n}^k} \right)$ w.r.t. $ {{\bf{q}}_{t,n}^k}$ can be given as $\nabla \tau_t\left({{{\bf{q}}_{t,n}^{k}}}\right)={\left[ {\frac{{\partial {\tau _t}\left( {{\bf{q}}_{t,n}^k} \right)}}{{\partial x_{t,n}^k}},\frac{{\partial {\tau _t}\left( {{\bf{q}}_{t,n}^k} \right)}}{{\partial y_{t,n}^k}}} \right]^T}$ and $\nabla^2 \tau_t\left({{{\bf{q}}_{t,n}^{k}}}\right)=\left[ {\begin{array}{*{20}{c}}
		{\frac{{{\partial ^2}{\tau _t}\left( {{\bf{q}}_{t,n}^k} \right)}}{{\partial x_{t,n}^k\partial x_{t,n}^k}}}&{\frac{{{\partial ^2}{\tau _t}\left( {{\bf{q}}_{t,n}^k} \right)}}{{\partial x_{t,n}^k\partial y_{t,n}^k}}}\\
		{\frac{{{\partial ^2}{\tau _t}\left( {{\bf{q}}_{t,n}^k} \right)}}{{\partial y_{t,n}^k\partial x_{t,n}^k}}}&{\frac{{{\partial ^2}{\tau _t}\left( {{\bf{q}}_{t,n}^k} \right)}}{{\partial y_{t,n}^k\partial y_{t,n}^k}}}
\end{array}} \right]$, whose elements are given by
\begin{align}
&\frac{{\partial {\tau _t}\left( {{\bf{q}}_{t,n}^k} \right)}}{{\partial x_{t,n}^k}} =  - \frac{{4\pi }}{\lambda }\sum\limits_{i = 1}^{{L_t}} {\left| {{\eta _i}} \right|\sin \theta _t^i\cos \phi _t^i\sin \left( {{\kappa _i}\left( {{\bf{q}}_{t,n}^k} \right)} \right)},\\
&\frac{{\partial {\tau _t}\left( {{\bf{q}}_{t,n}^k} \right)}}{{\partial x_{t,n}^k}} =  - \frac{{4\pi }}{\lambda }\sum\limits_{i = 1}^{{L_t}} {\left| {{\eta _i}} \right|\cos \theta _t^i\sin \left( {{\kappa _i}\left( {{\bf{q}}_{t,n}^k} \right)} \right)},
\end{align}
and
\begin{align}
	&\frac{{{\partial ^2}{\tau _t}\left( {{\bf{q}}_{t,n}^k} \right)}}{{\partial x_{t,n}^k\partial x_{t,n}^k}} =  - \frac{{8{\pi ^2}}}{{{\lambda ^2}}}\sum\limits_{i = 1}^{{L_t}} {\left| {{\eta _i}} \right|{{\sin }^2}\theta _t^i{{\cos }^2}\phi _t^i\cos \left( {{\kappa _i}\left( {{\bf{q}}_{t,n}^k} \right)} \right)},\\
	&\frac{{{\partial ^2}{\tau _t}\left( {{\bf{q}}_{t,n}^k} \right)}}{{\partial x_{t,n}^k\partial y_{t,n}^k}} =  - \frac{{8{\pi ^2}}}{{{\lambda ^2}}}\sum\limits_{i = 1}^{{L_t}} {\left| {{\eta _i}} \right|\sin \theta _t^i\cos \theta _t^i\cos \phi _t^i\cos \left( {{\kappa _i}\left( {{\bf{q}}_{t,n}^k} \right)} \right)} ,\\
	&\frac{{{\partial ^2}{\tau _t}\left( {{\bf{q}}_{t,n}^k} \right)}}{{\partial y_{t,n}^k\partial x_{t,n}^k}} =  - \frac{{8{\pi ^2}}}{{{\lambda ^2}}}\sum\limits_{i = 1}^{{L_t}} {\left| {{\eta _i}} \right|\sin \theta _t^i\cos \theta _t^i\cos \phi _t^i\cos \left( {{\kappa _i}\left( {{\bf{q}}_{t,n}^k} \right)} \right)} ,\\
	&\frac{{{\partial ^2}{\tau _t}\left( {{\bf{q}}_{t,n}^k} \right)}}{{\partial y_{t,n}^k\partial y_{t,n}^k}} =  - \frac{{8{\pi ^2}}}{{{\lambda ^2}}}\sum\limits_{i = 1}^{{L_t}} {\left| {{\eta _i}} \right|{{\cos }^2}\theta _t^i\cos \left( {{\kappa _i}\left( {{\bf{q}}_{t,n}^k} \right)} \right)}.
\end{align}

\section{Construction of $\delta_{t,n}^k$}
\label{Appendix6}
Since
\begin{equation}
	\begin{split}
		\left\| {{\nabla ^2}{\tau _t}\left( {{\bf{q}}_{t,n}^k} \right)} \right\|_2^2 \le& \left\| {{\nabla ^2}{\tau _t}\left( {{\bf{q}}_{t,n}^k} \right)} \right\|_F^2\\
		=& {\left( {\frac{{{\partial ^2}{\tau _t}\left( {{\bf{q}}_{t,n}^k} \right)}}{{\partial x_{t,n}^k\partial x_{t,n}^k}}} \right)^2} + {\left( {\frac{{{\partial ^2}{\tau _t}\left( {{\bf{q}}_{t,n}^k} \right)}}{{\partial x_{t,n}^k\partial y_{t,n}^k}}} \right)^2}\\
		&+ {\left( {\frac{{{\partial ^2}{\tau _t}\left( {{\bf{q}}_{t,n}^k} \right)}}{{\partial y_{t,n}^k\partial x_{t,n}^k}}} \right)^2} + {\left( {\frac{{{\partial ^2}{\tau _t}\left( {{\bf{q}}_{t,n}^k} \right)}}{{\partial y_{t,n}^k\partial y_{t,n}^k}}} \right)^2}\\
		\le& 4{\left( {\frac{{8{\pi ^2}}}{{{\lambda ^2}}}\sum\limits_{i = 1}^{{L_t}} {{{\left| {{\eta _i}} \right|}}} } \right)^2},
	\end{split}
\end{equation}
and $\left\| {{\nabla ^2}{\tau _t}\left( {{\bf{q}}_{t,n}^k} \right)} \right\|_2{\mathbf I}\succeq {{\nabla ^2}{\tau _t}\left( {{\bf{q}}_{t,n}^k} \right)}$, we can construct $\delta_{t,n}^k$ as
\begin{equation}\label{delta}
	\delta_{t,n}^k={\frac{{16{\pi ^2}}}{{{\lambda ^2}}}\sum\limits_{i = 1}^{{L_t}} {{{\left| {{\eta _i}} \right|}}} },
\end{equation}
which satisfies that
\begin{equation}
	\delta _{t,n}^k{\bf{I}} \succeq {\left\| {{\nabla ^2}{\tau _t}\left( {{\bf{q}}_{t,n}^k} \right)} \right\|_2}{\bf{I}} \succeq {\nabla ^2}{\tau _t}\left( {{\bf{q}}_{t,n}^k} \right).
\end{equation}

\ifCLASSOPTIONcaptionsoff
  \newpage
\fi

\bibliographystyle{IEEEtran}      
\bibliography{IEEEabrv,ref} 

@ARTICLE{6G1,
	author={Xiao, Yue and Ye, Ziqiang and Wu, Mingming and Li, Haoyun and Xiao, Ming and Alouini, Mohamed-Slim and Al-Hourani, Akram and Cioni, Stefano},
	journal={IEEE J. Sel. Areas Commun.}, 
	title={Space-Air-Ground Integrated Wireless Networks for {6G}: Basics, Key Technologies, and Future Trends}, 
	year={2024},
	month=dec,
	volume={42},
	number={12},
	pages={3327-3354},
	doi={10.1109/JSAC.2024.3492720}}

@ARTICLE{6G2,
	author={Yang, Ping and Xiao, Yue and Xiao, Ming and Li, Shaoqian},
	journal={IEEE Netw.}, 
	title={6{G} Wireless Communications: Vision and Potential Techniques}, 
	year={2019},
	month={Jul./Aug.},
	volume={33},
	number={4},
	pages={70-75},
	doi={10.1109/MNET.2019.1800418}}

@article{6G3,
	title={Fluid Antenna Systems Enabling {6G}: Principles, Applications, and Research Directions},
	author={Wu, Tuo and Zhi, Kangda and Yao, Junteng and Lai, Xiazhi and Zheng, Jianchao and Niu, Hong and Elkashlan, Maged and Wong, Kai-Kit and Chae, Chan-Byoung and Ding, Zhiguo and others},
	journal={arXiv preprint arXiv:2412.03839},
	year={2024}
}

@ARTICLE{6G4,
	author={New, Wee Kiat and Wong, Kai-Kit and Xu, Hao and Wang, Chao and Ghadi, Farshad Rostami and Zhang, Jichen and Rao, Junhui and Murch, Ross and Ramírez-Espinosa, Pablo and Morales-Jimenez, David and Chae, Chan-Byoung and Tong, Kin-Fai},
	journal={IEEE Commun. Surveys Tuts.}, 
	title={A Tutorial on Fluid Antenna System for {6G} Networks: Encompassing Communication Theory, Optimization Methods and Hardware Designs}, 
	year={2024},
	volume={},
	number={},
	note={early access},
	pages={1-1},
	doi={10.1109/COMST.2024.3498855}}

@ARTICLE{6G5,
	author={Yang, Zhaojie and Li, Yunye and Guan, Yong Liang and Fang, Yi},
	journal={IEEE J. Sel. Areas Commun.}, 
	title={Source-Constrained Hierarchical Modulation Systems With Protograph {LDPC} Codes: A Promising Transceiver Design for Future {6G}-Enabled {I}o{T}}, 
	year={2025},
	month=apr,
	volume={43},
	number={4},
	pages={1103-1117},
	doi={10.1109/JSAC.2025.3531540}}

@ARTICLE{analytical,
	author={Khammassi, Malek and Kammoun, Abla and Alouini, Mohamed-Slim},
	journal={IEEE Trans. Wireless Commun.}, 
	title={A New Analytical Approximation of the Fluid Antenna System Channel}, 
	year={2023},
	month=dec,
	volume={22},
	number={12},
	pages={8843-8858},
	doi={10.1109/TWC.2023.3266411}}

@ARTICLE{CE,
	author={Skouroumounis, Christodoulos and Krikidis, Ioannis},
	journal={IEEE Trans. Commun.}, 
	title={Fluid Antenna With Linear {MMSE} Channel Estimation for Large-Scale Cellular Networks}, 
	year={2023},
	month=feb,
	volume={71},
	number={2},
	pages={1112-1125},
	doi={10.1109/TCOMM.2022.3230861}}

@ARTICLE{FAS1,
	author={Wong, Kai-Kit and Shojaeifard, Arman and Tong, Kin-Fai and Zhang, Yangyang},
	journal={IEEE Trans. Wireless Commun.}, 
	title={Fluid Antenna Systems}, 
	year={2021},
	month=mar,
	volume={20},
	number={3},
	pages={1950-1962},
	doi={10.1109/TWC.2020.3037595}}

@ARTICLE{PLS1,
		author={Yao, Junteng and Xin, Liangxiao and Wu, Tuo and Jin, Ming and Wong, Kai-Kit and Yuen, Chau and Shin, Hyundong},
		journal={IEEE Internet Things J.}, 
		title={{FAS} for Secure and Covert Communications}, 
		year={2025},
		note={early access},
		volume={},
		number={},
		pages={1-1},
		doi={10.1109/JIOT.2025.3539737}}

@ARTICLE{PLS2,
	author={Xu, Hao and Wong, Kai-Kit and New, Wee Kiat and Li, Guyue and Ghadi, Farshad Rostami and Zhu, Yongxu and Jin, Shi and Chae, Chan-Byoung and Zhang, Yangyang},
	journal={IEEE Commun. Lett.}, 
	title={Coding-Enhanced Cooperative Jamming for Secret Communication in Fluid Antenna Systems}, 
	year={2024},
	month=sep,
	volume={28},
	number={9},
	pages={1991-1995},
	doi={10.1109/LCOMM.2024.3418338}}

@misc{PLS3,
	title={Secrecy Performance Analysis of {RIS}-Aided Fluid Antenna Systems}, 
	author={Farshad Rostami Ghadi and Kai-Kit Wong and Masoud Kaveh and F. Javier Lopez-Martinez and Wee Kiat New and Hao Xu},
	year={2024},
	eprint={2408.14969},
	archivePrefix={arXiv},
	primaryClass={cs.IT},
	url={https://arxiv.org/abs/2408.14969}, 
}

@misc{PLS4,
	title={Physical Layer Security over Fluid Antenna Systems: Secrecy Performance Analysis}, 
	author={Farshad Rostami Ghadi and Kai-Kit Wong and F. Javier Lopez-Martinez and Wee Kiat New and Hao Xu and Chan-Byoung Chae},
	year={2024},
	eprint={2402.05722},
	archivePrefix={arXiv},
	primaryClass={cs.IT},
	url={https://arxiv.org/abs/2402.05722}, 
}

@misc{PLS5,
	title={Physical Layer Security in {FAS}-aided Wireless Powered {NOMA} Systems}, 
	author={Farshad Rostami Ghadi and Masoud Kaveh and Kai-Kit Wong and Diego Martin and Riku Jantti and Zheng Yan},
	year={2025},
	eprint={2501.09106},
	archivePrefix={arXiv},
	primaryClass={cs.IT},
	url={https://arxiv.org/abs/2501.09106}, 
}

@ARTICLE{RIS1,
	author={Yao, Junteng and Zheng, Jianchao and Wu, Tuo and Jin, Ming and Yuen, Chau and Wong, Kai-Kit and Adachi, Fumiyuki},
	journal={IEEE Trans. Veh. Technol.}, 
	title={{FAS}-{RIS} Communication: Model, Analysis, and Optimization}, 
	year={2025},
	note={early access},
	volume={},
	number={},
	pages={1-6},
	doi={10.1109/TVT.2025.3537294}}

@misc{RIS2,
	title={A Framework of {FAS}-{RIS} Systems: Performance Analysis and Throughput Optimization}, 
	author={Junteng Yao and Xiazhi Lai and Kangda Zhi and Tuo Wu and Ming Jin and Cunhua Pan and Maged Elkashlan and Chau Yuen and Kai-Kit Wong},
	year={2024},
	eprint={2407.08141},
	archivePrefix={arXiv},
	primaryClass={eess.SP},
	url={https://arxiv.org/abs/2407.08141}, 
}

@misc{CR,
	title={{FAS}-Driven Spectrum Sensing for Cognitive Radio Networks}, 
	author={Junteng Yao and Ming Jin and Tuo Wu and Maged Elkashlan and Chau Yuen and Kai-Kit Wong and George K. Karagiannidis and Hyundong Shin},
	year={2024},
	eprint={2411.08383},
	archivePrefix={arXiv},
	primaryClass={eess.SP},
	url={https://arxiv.org/abs/2411.08383}, 
}

@misc{BackScatter,
	title={Performance Analysis of Fluid Antenna-aided Backscatter Communications Systems}, 
	author={Farshad Rostami Ghadi and Masoud Kaveh and Kai-Kit Wong},
	year={2024},
	eprint={2401.11820},
	archivePrefix={arXiv},
	primaryClass={cs.IT},
	url={https://arxiv.org/abs/2401.11820}, 
}

@misc{5GNR,
	title={Fluid Antenna System Empowering 5{G} {NR}}, 
	author={Hanjiang Hong and Kai-Kit Wong and Haoyang Li and Hao Xu and Han Xiao and Hyundong Shin and Kin-Fai Tong and Yangyang Zhang},
	year={2025},
	eprint={2503.05384},
	archivePrefix={arXiv},
	primaryClass={eess.SP},
	url={https://arxiv.org/abs/2503.05384}, 
}

@ARTICLE{ISAC1,
	author={Ye, Yuqi and You, Li and Xu, Hao and Elzanaty, Ahmed and Wong, Kai-Kit and Gao, Xiqi},
	journal={IEEE Trans. Veh. Technol.}, 
	title={{SCNR} Maximization for {MIMO} {ISAC} Assisted by Fluid Antenna System}, 
	year={2025},
	note={early access},
	volume={},
	number={},
	pages={1-6},
	doi={10.1109/TVT.2025.3557859}}

@misc{ISAC2,
	title={Shifting the {ISAC} Trade-Off with Fluid Antenna Systems}, 
	author={Jiaqi Zou and Hao Xu and Chao Wang and Lvxin Xu and Songlin Sun and Kaitao Meng and Christos Masouros and Kai-Kit Wong},
	year={2024},
	eprint={2405.05715},
	archivePrefix={arXiv},
	primaryClass={eess.SP},
	url={https://arxiv.org/abs/2405.05715}, 
}

@ARTICLE{ISAC3,
	author={Zhou, Liaoshi and Yao, Junteng and Jin, Ming and Wu, Tuo and Wong, Kai-Kit},
	journal={IEEE Wireless Commun. Lett.}, 
	title={Fluid Antenna-Assisted {ISAC} Systems}, 
	year={2024},
	month=dec,
	volume={13},
	number={12},
	pages={3533-3537},
	doi={10.1109/LWC.2024.3476148}}

@INPROCEEDINGS{motor1,
	author={Zhuravlev, Andrey and Razevig, Vladimir and Ivashov, Sergey and Bugaev, Alexander and Chizh, Margarita},
	booktitle={Proc. IEEE International Conf. Microwaves Commun. Antennas Electron. Syst. (COMCAS)}, 
	title={Experimental simulation of multi-static radar with a pair of separated movable antennas}, 
	year={2015},
	month=nov,
	address={Tel Aviv, Israel},
	volume={},
	number={},
	pages={1-5},
	doi={10.1109/COMCAS.2015.7360379}}

@misc{motor2,
	title={Historical Review of Fluid Antenna and Movable Antenna}, 
	author={Lipeng Zhu and Kai-Kit Wong},
	year={2024},
	eprint={2401.02362},
	archivePrefix={arXiv},
	primaryClass={cs.IT},
	url={https://arxiv.org/abs/2401.02362}, 
}

@misc{trajectory1,
	title={Minimizing Movement Delay for Movable Antennas via Trajectory Optimization}, 
	author={Qingliang Li and Weidong Mei and Boyu Ning and Rui Zhang},
	year={2024},
	eprint={2408.12813},
	archivePrefix={arXiv},
	primaryClass={eess.SY},
	url={https://arxiv.org/abs/2408.12813}, 
}

@article{trajectory2,
	title={{UAV}-Enabled Wireless Networks With Movable-Antenna Array: Flexible Beamforming and Trajectory Design},
	volume={14},
	ISSN={2162-2345},
	url={http://dx.doi.org/10.1109/LWC.2024.3451246},
	DOI={10.1109/lwc.2024.3451246},
	number={3},
	journal={IEEE Wireless Commun. Lett.},
	publisher={Institute of Electrical and Electronics Engineers (IEEE)},
	author={Liu, Wenchao and Zhang, Xuhui and Xing, Huijun and Ren, Jinke and Shen, Yanyan and Cui, Shuguang},
	year={2025},
	month=mar, pages={566–570} }

@ARTICLE{UAV,
	author={Zeng, Yong and Zhang, Rui},
	journal={IEEE Trans. Wireless Commun.}, 
	title={Energy-Efficient {UAV} Communication With Trajectory Optimization}, 
	year={2017},
	month=jun,
	volume={16},
	number={6},
	pages={3747-3760},
	doi={10.1109/TWC.2017.2688328}}

@ARTICLE{WMMSE1,
	author={Christensen, Søren Skovgaard and Agarwal, Rajiv and De Carvalho, Elisabeth and Cioffi, John M.},
	journal={IEEE Trans. Wireless Commun.}, 
	title={Weighted sum-rate maximization using weighted {MMSE} for {MIMO}-{BC} beamforming design}, 
	year={2008},
	month=dec,
	volume={7},
	number={12},
	pages={4792-4799},
	doi={10.1109/T-WC.2008.070851}}

@article{WMMSE2,
	title={Nonlinear programming},
	author={Bertsekas, Dimitri P},
	journal={Journal of the Operational Research Society},
	volume={48},
	number={3},
	pages={334--334},
	year={1997},
	publisher={Taylor \& Francis}
}

@ARTICLE{MM,
	author={Sun, Ying and Babu, Prabhu and Palomar, Daniel P.},
	journal={IEEE Trans. Signal Process.}, 
	title={Majorization-Minimization Algorithms in Signal Processing, Communications, and Machine Learning}, 
	year={2017},
	month=feb,
	volume={65},
	number={3},
	pages={794-816},
	doi={10.1109/TSP.2016.2601299}}

@ARTICLE{secure,
	author={Tang, Jun and Pan, Cunhua and Zhang, Yang and Ren, Hong and Wang, Kezhi},
	journal={IEEE Trans. Commun.}, 
	title={Secure {MIMO} Communication Relying on Movable Antennas}, 
	year={2024},
	volume={},
	number={},
	pages={1-1},
	doi={10.1109/TCOMM.2024.3465369}}

@book{QCQP,
	title={Lectures on modern convex optimization: analysis, algorithms, and engineering applications},
	author={Ben-Tal, Aharon and Nemirovski, Arkadi},
	year={2001},
	publisher={SIAM}
}

@ARTICLE{Capacity,
	author={Ma, Wenyan and Zhu, Lipeng and Zhang, Rui},
	journal={IEEE Trans. Wireless Commun.}, 
	title={{MIMO} Capacity Characterization for Movable Antenna Systems}, 
	year={2024},
	month=apr,
	volume={23},
	number={4},
	pages={3392-3407},
	doi={10.1109/TWC.2023.3307696}}

@ARTICLE{SM,
	author={Zhu, Jing and Chen, Gaojie and Gao, Pengyu and Xiao, Pei and Lin, Zihuai and Quddus, Atta Ul},
	journal={IEEE Trans. Wireless Commun.}, 
	title={Index Modulation for Fluid Antenna-Assisted {MIMO} Communications: System Design and Performance Analysis}, 
	year={2024},
	month=aug,
	volume={23},
	number={8},
	pages={9701-9713},
	doi={10.1109/TWC.2024.3364712}}

@article{LowerBound,
	author = {Junxiao Song and Babu, Prabhu and Palomar, Daniel P.},
	title = {Optimization Methods for Designing Sequences With Low Autocorrelation Sidelobes},
	journal={IEEE Trans. Signal Process.},
	year = {2015},
	month=aug,
	issue_date = {Aug.1, 2015},
	publisher = {IEEE Press},
	volume = {63},
	number = {15},
	issn = {1053-587X},
	doi = {10.1109/TSP.2015.2425808},
	month = aug,
	pages = {3998–4009},
	numpages = {12}
}

@ARTICLE{6DMA,
	author={Shao, Xiaodan and Jiang, Qijun and Zhang, Rui},
	journal={IEEE Trans. Wireless Commun.}, 
	title={6{D} Movable Antenna Based on User Distribution: Modeling and Optimization}, 
	year={2025},
	month=jan,
	volume={24},
	number={1},
	pages={355-370},
	doi={10.1109/TWC.2024.3492195}}

\end{document}